\pdfoutput=1
\documentclass[pdftex,pra,aps,onecolumn,twoside,nofootinbib,showkeys,showpacs]{revtex4}

\usepackage{float}
\usepackage{indentfirst}
\usepackage{url}
\usepackage[stable]{footmisc}
\usepackage{fancyhdr}
\usepackage{appendix}
\usepackage[all]{xy}
\usepackage{color}
 \ifpdf
\usepackage[pdftex]{graphicx}
 \else
\usepackage{graphicx}
 \fi
\usepackage{youngtab}
\usepackage[T1]{fontenc}
\usepackage{textcomp}
\usepackage{newcent}
\usepackage[sc,noBBpl]{mathpazo}
\usepackage{amsmath,amsfonts,amssymb,amsthm}
\usepackage[mathscr]{eucal}
\usepackage{enumerate}
\theoremstyle{plain}
\newtheorem{theorem}{Theorem}
\newtheorem{lemma}[theorem]{Lemma}
\newtheorem{proposition}[theorem]{Proposition}
\newtheorem{corollary}[theorem]{Corollary}

\newtheorem{definition}[theorem]{Definition}
\newtheorem{fact}[theorem]{Fact}
\newtheorem{notation}[theorem]{Notation}
\newtheoremstyle{note}{\topsep}{\topsep}{\slshape}{}{\scshape}{}{ }{}
\theoremstyle{note}

\newtheorem{remark}[theorem]{Remark}

\newcommand\Span{\operatorname{span}}

\newcommand\ind{\operatorname{ind}}
\newcommand{\mbV}{\mathbb{V}}

\newcommand{\mbI}{\text{\noindent\(\mathds{1}\)}}

\usepackage{dsfont}
\newcommand\field{\mathbb}

\newcommand\C{\field{C}}

\newcommand\rank{\operatorname{rank}}

\newcommand\tr{\operatorname{Tr}}

\newcommand\id{\operatorname{\mathrm{Id}}}
\newcommand{\opV}{\operatorname{V}}

\newcommand{\<}{\langle}
\renewcommand{\>}{\rangle}

\newcommand\be{\begin{equation}}
\newcommand\ee{\end{equation}}
\newcommand\bea{\begin{array}}
\newcommand\eea{\end{array}}
\newcommand\ben{\begin{eqnarray}}
\newcommand\een{\end{eqnarray}}
\newcommand\ot{\otimes}
\newcommand\bei{\begin{itemize}}
\newcommand\eei{\end{itemize}}
\newcommand\bee{\begin{enumerate}}
\newcommand\eee{\end{enumerate}}

\begin{document}
\title{Group representation approach to $1 \rightarrow N$ universal quantum cloning machines}

\author{Micha{\l} Studzi\'nski$^{1,2}$, Piotr \'Cwikli\'nski$^{1,2}$, Micha{\l} Horodecki$^{1,2}$ and Marek Mozrzymas$^3$}
\affiliation{
$^1$ Institute of Theoretical Physics and Astrophysics, University of Gda\'nsk, 80-952 Gda\'nsk, Poland \\
$^2$ National Quantum Information Centre of Gda\'nsk, 81-824 Sopot, Poland \\
$^3$ Institute for Theoretical Physics, University of Wroc{\l}aw, 50-204 Wroc{\l}aw, Poland}

\date{\today}

\begin{abstract}
In this work, we revisit the problem of finding an admissible region of fidelities obtained after an application of an arbitrary $1 \rightarrow  N$ universal quantum cloner which has been recently solved in [A. Kay et al., Quant. Inf. Comput 13, 880 (2013)] from the side of cloning machines. Using group-theory formalism, we show that the allowed region for fidelities can be alternatively expressed in terms of overlaps of pure states with recently found irreducible representations of the commutant $U \ot U \ot \ldots \ot U \ot U^*$, which gives the characterization of the allowed region where states being cloned are figure of merit. Additionally, it is sufficient to take pure states with real coefficients only, which makes calculations simpler. To obtain the allowed region, we make a convex hull of possible ranges of fidelities related to a given irrep. Subsequently, two cases: $1 \rightarrow 2$ and $1 \rightarrow 3$ cloners, are studied for different dimensions of states as illustrative examples.
\end{abstract}

\pacs{03.67.-a, 03.65.Fd, 03.67.Hk}
\keywords{quantum cloning; asymmetric quantum cloning; qudit; representation theory}

\maketitle

\section{Introduction}
\label{sec:introduction}
A basic feature of entanglement is that contrary to classical correlations, it is monogamous. For example, if there is maximal entanglement between two parties, then no other party can be entangled with those parties. More generally, if $A$ is entangled with $B$ and $C$, then the entanglement must be considerably weaker. This phenomenon gives rise to the fact that quantum information cannot be copied, in contrast with information from the 'classical world'. In other words, one is not able to copy perfectly an arbitrary quantum state. In terms of monogamy, if one wants to prepare some number of copies of the initially unknown quantum state, fidelities of cloning cannot be all equal to 1, there is a trade-off. This basic feature is known as 'no-cloning theorem' and was recognized by Wootters and \.Zurek \cite{WoottersZurek}, and Dieks \cite{Dieks}.

On the other hand, copying is possible, but the quality of the copy can be very bad sometimes. That's why, the goal of finding the ultimate bounds for the quality of copying is an important task. A big effort has been made to solve it, starting from the work of Hillery and Bu\u{z}ek \cite{BuzekHillery}. In general, the subject was studied intensively, both for symmetric (all fidelities are equal) Universal Quantum Cloning Machines ($UQCM$) \cite{Bruss-cloning1998, Massar, Werner-cloning1998, KW, Wang-cloning2011}, and asymmetric (unequal fidelities) $UQCM$ \cite{Braunstein-cloning2001, Cerf-cloning2000, Fiurasek-cloning2005, Iblisdir-cloning2004, Iblisdir-cloning2005, Wang-cloning2011, Kay2009-cloning, Kaszlikowski2012-cloning, Yu2010-cloning}. See also \cite{Gisin, Fan2013-cloning} for reviews.
Nevertheless, for a long time there was a 'gap' in studies of quantum cloning - there was no general results on an admissible region of fidelities for universal asymmetric $1 \rightarrow N$ quantum cloning machines. The problem has been solved just recently in a series of papers \cite{Kay2009-cloning, Kaszlikowski2012-cloning} from the point of cloning machines. In \cite{Cwiklinski2012-cloning} the problem, for qubits, has been revisited using a group representation approach, namely Schur - Weyl duality, where the authors characterized the problem from the side of a cloned state and obtain that regions for fidelities can be obtained from plain and basic calculations of overlaps of pure quantum states with irreps of a symmetric group $S(n)$.

In this Letter, we shall consider a $1 \rightarrow N$ quantum cloning machine for qudits. Our task it to obtain an admissible region of fidelities after an application of that $UQCM$. In \cite{Cwiklinski2012-cloning}, it has been shown that it is possible to solve the problem for qubits using Schur-Weyl duality. Unfortunately it works only for that dimension of states and there is no way to extend it to higher dimensions by the usage of that dualism. Motivated by this, we turn our attention to, recently developed systematic method - decomposition of partially transposed permutation operators into its irreducible components~\cite{Studzinski2013,mozrzymas2013}, which allows to omit severe restrictions for the dimensions of states that has appeared previously. However, some modifications are necessary first, so the method suits our problem of cloning machines. We want to stress that to our best knowledge, it is also the first systematic application of that algebra in physics, and particulary - quantum information (see, \cite{EggelingWerner, Eggeling}, for the examples of some limited applications).

This work is organized as follows. In Section \ref{sec:problem}, we formulate our main problem: which values of fidelities are allowed after applying a $1 \rightarrow N$ quantum cloning machine for qudits. First, we reformulate the cloning problem in term of entanglement sharing and recall that a cloning fidelity can be connected with a singlet fraction value. Then, we point out that the strategy used in \cite{Cwiklinski2012-cloning} to solve a $1 \rightarrow N$ $UQCM$ for qubits is insufficient when one deals with higher dimensions of states $d$, ($d>2$), since using Schur-Weyl duality, one is not able to find a maximally entangled state that is invariant under $U \ot U$ transformations, the only thing that is known is the invariance under $U^{*} \ot U$ ones. That's why, the commutant structure of $U^{*} \ot U \ot \ldots \ot U$ is needed instead of that known from Schur-Weyl duality: $U \ot U \ot \ldots \ot U$. In Section \ref{sec:tools}, mathematical tools from \cite{Studzinski2013} that are necessary to solve the problem are very briefly mentioned, namely, examples of irreducible representations that are needed in our case study problems: $1 \rightarrow 2$ and $1 \rightarrow 3$ $UQCM$. Then, we proceed in Section \ref{sec:method} with showing how to connect method of calculations of the admissible region of fidelities from \cite{Cwiklinski2012-cloning} with mathematical tools from the previous section. It allows us to present in Section \ref{sec:main} the regions (focusing mainly on our examples $1 \rightarrow 2$ and $1 \rightarrow 3$ machines) that are allowed in the problem of $1 \rightarrow N$ cloning. Up to our best knowledge, it is the first graphical presentation of allowed regions for $d>2$. At the end, we compare our results in Section \ref{sec:com} with those obtained in \cite{Werner-cloning1998}, where results for symmetric cloning has been presented and from \cite{Kay2009-cloning, Kaszlikowski2012-cloning}, where the same problem as ours have been solved, but cloning machines were figures of merit. We obtain matching of results in both cases.

\section{Formulation and solution to the problem}
\label{sec:problem}
\subsection{Background of the problem}
Suppose that one has a universal cloning machine that produces clones with cloning fidelities $f_{1k}$, where $k \in {2, 3, \ldots, n}$ and the general, admissible region of fidelities is the figure of merit. The question that one can ask is the following:\\
\textit{Which values of cloning fidelities
$\left( f_{12}, f_{13}, \ldots, f_{1n} \right)$ are allowed for a (qudit) universal cloning machine?}

But since quantum cloning can be recast in a picture where one wants to share entanglement between some number of parties (see, for example, \cite{Cwiklinski2012-cloning, MHPH}). Therefor, we can equivalently state our problems in this formalism, where one evaluates singlet fractions $F_{1i}$ between the initial state and one of the copies. This allows to restate our question as:\\
\textit{Which values of n-tuples of singlet fractions
$\left( F_{12}, F_{13}, \ldots, F_{1n} \right)$ are allowed for an arbitrary state
of a maximally mixed first subsystem?}

\textbf{Remark:} Since these two quantities, cloning fidelities and singlet fractions, are connected \cite{MHPH}, in the next section we will adapt the term "fidelities" for the latter.

Let us now consider in more details the relation between cloning fidelities $f$ and the fidelities (singlet fractions) $F$.

Suppose that we are given with the maximally entangled qudit state
\be
| \psi^{+} \> = \frac{1}{\sqrt{d}} \sum_{i=1}^{d} |ii\>,
\ee
and we apply the $1\to N$ cloning machine~\footnote{described by a completely positive, trace preserving map $\widetilde{\Lambda}$} $\mathcal{CM}$ to the second subsystem of the $|\psi^+\>$, when the first is untouched. As a result we obtain $N+1$-partite mixed state that possesses all information about the cloning map $\widetilde{\Lambda}$. The state is of the form

\be
\rho_{1\ldots n} = \left( \operatorname{\mbI} \ot \widetilde{\Lambda} \right)\left( | \psi^{+} \> \< \psi^{+} | \right), \label{map}
\ee
where $n=N+1$, so that the index $i=1$ is related to an initial state, and $i=2,\ldots N+1$ are related to
clones. The fidelities of clones are strictly related to fidelities of reduced states $\rho_{1k}$ with maximally entangled state \cite{MHPH}:
\be
f_i= \frac{F_id+1}{d+1}.
\ee
Here $f_i=\<\psi_{in}|\rho_{out}^i|\psi_{in}\>$ is fidelity of $i$-th clone where $\<\ldots\>$
is the uniform average over an input state $\psi_{in}$, and $F_i=\<\psi_+|\rho_{1,i} |\psi_+\>$.

An allowed region for quantum cloning, can be calculated then by evaluating singlet fractions $F_{1i}$ between the initial state and one of the copies, denoted by
% \ben
% &&F_{12} = \< \psi_{12}^{+} | \tr_{34} (\rho_{1234}) | \psi_{12}^{+} \>, \nonumber \\
% &&F_{13} = \< \psi_{13}^{+} | \tr_{24} (\rho_{1234}) | \psi_{13}^{+} \>, \nonumber \\
% &&F_{14} = \< \psi_{14}^{+} | \tr_{23} (\rho_{1234}) | \psi_{14}^{+} \>.
% \een
% Which can be also written as:
\be
F_{1i}=\< \psi^{+}_{1i} | \tr_{\overline{1i}}(\rho_{1\ldots n}) | \psi^{+}_{1i} \> \ \text{or} \  F_{1i} = \< \psi^{-}_{1i} | \tr_{\overline{1i}}(\widetilde{\rho}_{1\ldots n}) | \psi^{-}_{1i} \>
\label{note},
\ee
where $1< i\leq n$, $\tr_{\overline{1i}}$ means partial trace over all systems except $1i$,
and $|\psi^{-}_{1i}\>$ and $\widetilde{\rho}_{1\ldots n}$ are defined below.

Let us show here, why we have been able to use Schur-Weyl duality and commutant structure of $U^{\ot n}$ for qubits cloning machines \cite{Cwiklinski2012-cloning} and explain why it does not work for higher dimensions of states ($d>2$). For qudits, in principle, the vector $|\psi^{-}_{1i}\> = U \ot \operatorname{\mbI} | \psi^{+}_{1 \widetilde{1}} \>$, $| \psi^{-} \>$ needs to be obtain after an application of $U$. For qubits, one can use Bell states $|\psi^{+}\> = \frac{1}{\sqrt{2}}(|00\>+|11\>)$ and $|\psi^{-}\> = \frac{1}{\sqrt{2}}(|01\>-|10\>)$ and show that the vector $|\psi^{-}\>$ is obtained after the action of the Pauli matrix $-i\sigma_{y}$ on $|\psi^{+}\>$)
Using that we can write
\be |\psi^{-}_{1 \widetilde{1}} \> = U \ot \operatorname{\mbI} |\psi^{+}_{1\widetilde{1}} \>, \label{EqU} \ee
where $U = -i \sigma_y$.
The state $\widetilde{\rho}_{1234}$ from equation.~\eqref{note} is obtained after the following transformation:
\be
\begin{split}
\label{map1}
\widetilde{\rho}_{1\ldots n} &= ( \operatorname{\mbI} \ot \widetilde{\Lambda} )|\psi^{-}_{1\widetilde{1}} \>\<\psi^{-}_{1\widetilde{1}}| \\ &=
( U \ot \operatorname{\mbI} )  \left(( \operatorname{\mbI} \ot \widetilde{\Lambda} )|\psi^{+}_{1\widetilde{1}}\>   \< \psi^{+}_{1\widetilde{1}} | \right) (U \ot \operatorname{\mbI})^{\dag}.
\end{split}
\ee
The $n-$partites states $\widetilde{\rho}_{1\ldots n}$, with the constraint $\widetilde{\rho}_1=\mbI/2$,
are in one-to-one correspondence with cloning machines.

However, now the problem is formulated in terms of singlet fractions with states $|\psi^-\>$ rather than $|\psi^+\>$. The former states are invariant under $U\ot U$
transformation for any $U$. Therefore to obtain the region of fidelities with $|\psi^-\>$ states
it is enough to consider states $\rho_{1\ldots n}$ that are invariant under $U^{\ot n}$
transformations. There exists well known formalism that allows to deal with states possessing such symmetry, called Schur-Weyl duality that combines representation theory for unitary group with that of  group of permutations. We have successfully applied this formalism in \cite{Cwiklinski2012-cloning}. However in dimensions $d>2$ there is no maximally entangled state, that would be $U\ot U$-invariant. Therefore, the Schur-Weyl formalism cannot be used.

Instead, it is known, that the state $|\psi^+\>$ is $U^*\ot U$ invariant \cite{Horodecki_2001}, hence we should consider $U^*\ot U^{\ot n-1}$ invariant states. The formalism, related to this kind of symmetry is not so well developed as the previous one, and there are quite basic differences between the two.
In particular, while the representation of  $U^{\ot n}$ is dual to representation of another group - the symmetric group,
it is not the case for $U^*\ot U^{\ot n-1}$ which is dual to representation of an algebra, that does not satisfy
group axioms - an instance of so called Brauer algebra. While some general results concerning this type of algebras
have been known in literature (see, for example, \cite{EggelingWerner, Eggeling, Werner2006}), it has not been described in depth, in contrast to Schur-Weyl theory.
In particular, the explicit form of matrix elements of representations of the algebra, have been provided only recently
in \cite{Studzinski2013,mozrzymas2013}. In the following we solve the cloning problem applying these new tools.

\subsection{Mathematical tools}
\label{sec:tools}
As it was said before, to solve our problem, the knowledge of irreducible representations of a $U^* \ot U \ot \ldots \ot U$ case is necessary. In a recent papers \cite{Studzinski2013,mozrzymas2013} this problem has been addressed, so we can use the formalism presented there~\footnote{See also Appendix~\ref{remind} for a short review on this topic.}.

In the articles, the authors presented irreducible representations of partially transposed permutation operators $\opV^{t_n}(\sigma)$, where $\sigma\in S(n)$ and $t_n$ denotes partial transposition over the last subsystem. In our approach, we need similar results for irreps when partial transposition is taken over the first subsystem, i.e. we need irreps of $\opV^{t_1}(1k)$, where $1\leq k \leq n$ for $U^* \ot U \ot \ldots \ot U$ instead of $U \ot \ldots \ot U \ot U^*$. That's why, first, some work needs to be done to adapt the results, so they suit our problem. One can see that to obtain correct results, we have to take irreps for permutations in the form $(in)$, where $1\leq i \leq n-1$, i.e. we have the following mapping
\be
\label{mapping}
(12) \mapsto (1n), \ (13) \mapsto (2n) \ ,\ldots, \ (1n) \mapsto (n-1n).
\ee
In the next sections, for the simplicity, we introduce the notation that $t_n \equiv \ '$.
Now we are ready to present all irreps that are essential for our paper (case study examples). Of course our method works efficiently for an arbitrary number of particles $n$ and dimensions of Hilbert space $d$, but here we present them only for $n=3,4$, because for these cases we are able to represent our results graphically.
\begin{itemize}
\item Case when $n=3$. In this case in algebra $\mathcal{M}$ we have only one irrep labeled by trivial partition $\alpha=(1)$.
\be
\mbV'_{\alpha}(13)=\frac{1}{2}\left(
\begin{array}{cc}
d+1 & -\sqrt{d^{2}-1} \\
-\sqrt{d^{2}-1} & d-1%
\end{array}%
\right) ,~\mbV'_{\alpha}(23)=\frac{1}{2}\left(
\begin{array}{cc}
d+1 & \sqrt{d^{2}-1} \\
\sqrt{d^{2}-1} & d-1%
\end{array}%
\right)
\ee
% \be
% \label{exIrrepsn3}
% \widetilde{\mbV}'_{\lambda}(13)=\frac{1}{2}\begin{bmatrix}d+\sqrt{d^2-1} & 1 \\1 & d-\sqrt{d^2-1} \end{bmatrix}, \quad \widetilde{\mbV}'_{\lambda}(23)=\frac{1}{2}\begin{bmatrix} d-\sqrt{d^2-1} & 1\\ 1 & d+\sqrt{d^2-1} \end{bmatrix}.
% \ee
\item Case when $n=4$. In this case in algebra $\mathcal{M}$ we have two irreps labeled by partitions $\alpha_1=(2)$ and $\alpha_2=(1,1)$. For partition $\alpha_1$ we deal with matrices 3x3 for any $d\geq 1$:
\be
\label{alpha1}
\begin{split}
\mbV'_{\alpha_1}(14)&=\frac{1}{3}D^{\alpha_1}\left(
\begin{array}{ccc}
\frac{1}{6} & \frac{-1}{2\sqrt{3}} & \frac{1}{3\sqrt{2}} \\
\frac{-1}{2\sqrt{3}} & \frac{1}{2} & \frac{-1}{\sqrt{6}} \\
\frac{1}{3\sqrt{2}} & \frac{-1}{\sqrt{6}} & \frac{1}{3}
\end{array}%
\right) D^{\alpha_1},\quad \mbV'_{\alpha_1}(24)=\frac{1}{3}D^{\alpha_1}\left(
\begin{array}{ccc}
\frac{1}{6} & \frac{1}{2\sqrt{3}} & \frac{1}{3\sqrt{2}} \\
\frac{1}{2\sqrt{3}} & \frac{1}{2} & \frac{1}{\sqrt{6}} \\
\frac{1}{3\sqrt{2}} & \frac{1}{\sqrt{6}} & \frac{1}{3}%
\end{array}%
\right) D^{\alpha_1},\\
\mbV'_{\alpha_1}(34)&=\frac{1}{3}D^{\alpha_1}\left(
\begin{array}{ccc}
\frac{2}{3} & 0 & \frac{-2}{3\sqrt{2}} \\
0 & 0 & 0 \\
\frac{-2}{3\sqrt{2}} & 0 & \frac{1}{3}%
\end{array}%
\right) D^{\alpha_1},
\end{split}
\ee
where
\be
D^{\alpha_1}=\left(
\begin{array}{ccc}
\sqrt{d-1} & 0 & 0 \\
0 & \sqrt{d-1} & 0 \\
0 & 0 & \sqrt{d+2}%
\end{array}%
\right)
\ee
and $\varepsilon ^{2}=1$. For partition $\alpha_2$ situation is more complicated. Dimension of irrep $\alpha_2$ depends on dimension of local Hilbert space $d$. Namely for any $d \geq 3$ we have
\be
\begin{split}
\mbV'_{\alpha_2}(14)&=\frac{1}{3}D^{\alpha_2}\left(
\begin{array}{ccc}
\frac{1}{2} & \frac{-1}{2\sqrt{3}} & \frac{-1}{\sqrt{6}} \\
\frac{-1}{2\sqrt{3}} & \frac{1}{6} & \frac{1}{3\sqrt{2}} \\
\frac{-1}{\sqrt{6}} & \frac{1}{3\sqrt{2}} & \frac{1}{3}%
\end{array}%
\right) D^{\alpha_2},\quad \mbV'_{\alpha_2}(24))=\frac{1}{3}D^{\alpha_2}\left(
\begin{array}{ccc}
\frac{1}{2} & \frac{1}{2\sqrt{3}} & \frac{1}{\sqrt{6}} \\
\frac{1}{2\sqrt{3}} & \frac{1}{6} & \frac{1}{3\sqrt{2}} \\
\frac{1}{\sqrt{6}} & \frac{1}{3\sqrt{2}} & \frac{1}{3}%
\end{array}%
\right) D^{\alpha_2},\\
\mbV'_{\alpha_2}(34)&=\frac{1}{3}D^{\alpha_2}\left(
\begin{array}{ccc}
0 & 0 & 0 \\
0 & \frac{2}{3} & \frac{-\sqrt{2}}{3} \\
0 & \frac{-\sqrt{2}}{3} & \frac{1}{3}%
\end{array}%
\right)D^{\alpha_2},
\end{split}
\ee
where
\be
D^{\alpha_2}=\left(
\begin{array}{ccc}
\sqrt{d+1} & 0 & 0 \\
0 & \sqrt{d+1} & 0 \\
0 & 0 & \sqrt{d-2}%
\end{array}%
\right).
\ee
% \be
% \label{lambda1}
% \begin{split}
% \widetilde{\mbV}'_{\lambda_1}(14)&=\frac{1}{9}\begin{pmatrix}-2+5d+4\Delta_1 & 4-d+\Delta_1 & 4-d+\Delta_1 \\ 4-d+\Delta_1 & 1+2d-2\Delta_1 & 1+2d-2\Delta_1 \\ 4-d+\Delta_1 & 2+2d-2\Delta_1 & 1+2d-2\Delta_1 \end{pmatrix},\\
%  \widetilde{\mbV}'_{\lambda_1}(24)&=\frac{1}{9}\begin{pmatrix}1+2d-2\Delta_1 & 4-d+\Delta_1 & 1+2d-2\Delta_1 \\ 4-d+\Delta_1 & 5d-2+4\Delta_1 & 4-d+\Delta_1 \\ 1+2d-2\Delta_1 & 4-d+\Delta_1 & 2d+1-2\Delta_1 \end{pmatrix},\\
%  \widetilde{\mbV}'_{\lambda_1}(34)&=\frac{1}{9}\begin{pmatrix}1+2d-2\Delta_1 & 1+2d-2\Delta_1 & 4-d+\Delta_1 \\ 1+2d-2\Delta _1& 1+2d-2\Delta_1 & 4-d+\Delta_1 \\ 4-d+\Delta_1 & 4-d+\Delta_1 & -2+5d+4\Delta_1 \end{pmatrix},
% \end{split}
% \ee
% where $\Delta_1=\sqrt{(d-1)(d+2)}$.
For every $d<3$ (in our case only $d=2$ is interesting) we deal with matrices 2x2:
\be
\begin{split}
\mbV'_{\alpha_2}(14)=3\left(
\begin{array}{cc}
\frac{1}{2} & \frac{-1}{2\sqrt{3}} \\
\frac{-1}{2\sqrt{3}} & \frac{1}{6}%
\end{array}%
\right) ,\quad \mbV'_{\alpha_2}(24)=3\left(
\begin{array}{cc}
\frac{1}{2} & \frac{1}{2\sqrt{3}} \\
\frac{1}{2\sqrt{3}} & \frac{1}{6}%
\end{array}%
\right) ,\quad \mbV'_{\alpha_2}(34)=3\left(
\begin{array}{cc}
0 & 0 \\
0 & \frac{2}{3}%
\end{array}%
\right).
\end{split}
\ee
The full knowledge about irreps of  $\opV'(\sigma_{ab})$, where $\sigma_{ab}\in S(n)$ (see Notation~\ref{not9} in the Appendix~\ref{remind}) allows us to decompose these operators and density operators $\rho_{1\ldots n}$ which are $U^* \ot U \ot \ldots U $ invariant into block diagonal form
\be
\label{decomp1}
\opV'(\sigma_{ab})=\bigoplus_{\alpha} \mbI_{r(\alpha)} \ot \mbV'_{\alpha}(\sigma_{ab}), \quad
\rho_{1\ldots n}=\bigoplus_{\alpha} \mbI_{r(\alpha)} \ot \widetilde{\rho}^{\alpha},
\ee
where the direct sum runs over all inequivalent irreps $\alpha$, $r(\alpha)$ denotes the dimension of irrep $\alpha$ and $\widetilde{\rho}^{\alpha}$ is a representation of operator $\rho_{1\ldots n}$ on irrep $\alpha$. In the next paragraph we present how to use the decomposition from formula~\eqref{decomp1} and explicit matrix form of irreps of $\opV'(\sigma_{ab})$ to calculate fidelities.
% \be
% \label{dsm3}
% \widetilde{\mbV}'_{\lambda_2}(14)=\frac{1}{2}\begin{pmatrix}2+\sqrt{3} & 1\\ 1 & 2-\sqrt{3}\end{pmatrix},\quad  \widetilde{\mbV}'_{\lambda_2}(24)=\frac{1}{2}\begin{pmatrix}2-\sqrt{3} & 1\\ 1 & 2+\sqrt{3}\end{pmatrix},\quad   \widetilde{\mbV}'_{\lambda_2}(34)=\frac{1}{2}\begin{pmatrix}1 & -1\\ -1 & 1\end{pmatrix},
% \ee

% \be
% \label{dgr3}
% \begin{split}
% \widetilde{\mbV}'_{\lambda_2}(14)&=\frac{1}{9}\begin{pmatrix} 2+5d+4\Delta_2 & -4-d+\Delta_2 & 4+d-\Delta_2 \\ -4-d+\Delta_2 & -1+2d-2\Delta_2 & 1-2d+2\Delta_2 \\ 4+d-\Delta_2 & 1-2d+2\Delta_2 & -1+2d-2\Delta_2 \end{pmatrix},\\
%  \widetilde{\mbV}'_{\lambda_2}(24)&=\frac{1}{9}\begin{pmatrix}-1+2d-2\Delta_2 & -4-d+\Delta_2 & 1-2d+2\Delta_2 \\ -4-d+\Delta_2 & 2+5d+4\Delta_2 & 4+d-\Delta_2 \\ 1-2d-2\Delta_2 & 4+d-\Delta_2 & -1+2d-2\Delta_2 \end{pmatrix},\\
%  \widetilde{\mbV}'_{\lambda_2}(34)&=\frac{1}{9}\begin{pmatrix}-1+2d-2\Delta_2 & -1+2d-2\Delta_2 & 4+d-\Delta_2 \\ -1+2d-2\Delta _2 & -1+2d-2\Delta_2 & 4+d-\Delta_2 \\ 4+d-\Delta_2 & 4+d-\Delta_2 & 2+5d+4\Delta_2 \end{pmatrix},
% \end{split}
% \ee
% where $\Delta_2=\sqrt{(d+1)(d-2)}$.
\end{itemize}

\subsection{Method of calculations}
\label{sec:method}
Since, in principle, calculations techniques are similar to those from \cite{Cwiklinski2012-cloning}, in most cases, proofs are skipped and unless specified otherwise, we refer to the above-mentioned work for them.

In this section we provide a general formula for an allowed region of $N$-tuples
of fidelities in terms of overlaps of pure states with irreducible representations from the previous section.
This is contained in Theorem \ref{thm:main}.

\begin{lemma}
\label{FF}
Fidelity $F_{1k}$ as defined in \eqref{note} is of the form
\be
\label{FFF}
F_{1k}=\sum_{\alpha}F_{1k}^{\alpha},
\ee
where
\be
\label{ffv}
F_{1k}^{\alpha}
= \frac{1}{d} \tr \left( \rho^{\alpha} \mbV'_{\alpha}(k-1n) \right),
\ee
the index $(k-1n)$ means a permutation that swaps $k-1$ and $n$,
and $\rho^{\alpha}$'s are arbitrary normalized states on partition $\alpha$.
\end{lemma}

Again, from papers \cite{Studzinski2013,mozrzymas2013} we know that algebra of partially transposed permutation operators $\mathcal{A}'_n(d)$ splits into sum of two ideals, i.e. we have $\mathcal{A}'_n(d)=\mathcal{M}\oplus \mathcal{N}$. In Lemma~\ref{FF} we derived formulas for fidelities for elements in ideal $\mathcal{M}$, now we give similar formulas for elements in ideal $\mathcal{N}$. Physically it means that we looking for fidelities between maximally entangled state and some product state between input state and clones.
\begin{fact}
\label{fact}
Fidelity $F^{\mathcal{N}}_{1k}$ between state $|\psi_{1k}\>$ and a product state $\rho_{1k}=\frac{1}{d}\tr_{\overline{1k}}\left(\text{\noindent\(\mathds{1}\)}_1 \ot \rho_{2\ldots n}\right)$ is equal~\footnote{By $\tr_{\overline{1k}}$ we denote partial trace over all subsystems except $1^{\text{st}}$ and $k^{\text{th}}$.} to $1/d$.
\end{fact}

 Now we are in position to formulate the main theorem of this section:
\begin{theorem}
\label{thm:main}
The set $\mathcal{F}$ of admissible vectors of fidelities $\left\{ F_{12}, \ldots, F_{1n} \right\}$ is of the form
\be
\mathcal{F} = \operatorname{conv} \left( \bigcup_{\alpha}  \mathcal{F}^{\alpha} \right),
 \ee
where $\operatorname{conv}$ stands for a convex hull, the union runs over all  irreps and
\be  \mathcal{F}^{\alpha} = \left\{ \left( F_{12}^{\alpha},\ldots, F_{1n}^{\alpha}\right) \ : |\psi\rangle \in \C^{d_{\alpha}}\  \right\}, \ee \label{maintheo}
where $F_{1k}^{\alpha}$ are of the form: $F_{1k}^{\alpha} = \frac{1}{d}\< \psi |\mbV'_{\alpha}(k-1n) |\psi \>$, and where  $|\psi \>$ is a pure state.
\end{theorem}

Let us note that to determine the allowed region of fidelities, it is enough to consider only
vectors of real coefficients.
\begin{lemma}
\label{real}
To generate a convex hull of the allowed region of fidelities, it is sufficient to consider pure states of real coefficients only.
\end{lemma}

\subsection{Main result}
\label{sec:main}
In this section we present our results for two particular cases $1 \rightarrow 2$ and $1 \rightarrow 3$ universal quantum cloners.

Let us start with noting that to obtain a general answer to our question from Section \ref{sec:problem}, we need to have a mixture of all fidelities connected with our irreps: $\sum_{\alpha} p_{\alpha} F_{1N}^{\alpha}$. This implies that a convex hull is needed. On Figures \ref{fig:ch2} and \ref{fig:ch} we show plots for $N = 2,3$ and different dimensions $d$ before taking the convex hull, so one can see a contribution from each irrep. Then, we take one particular case, namely $1 \rightarrow N$ $UQCM$ and $d=3$ and present the convex hull for it that reproduces the allowed region for fidelities  (Figure \ref{fig:ch1}). All plots are obtained using $Mathematica$ software.

\begin{remark}
Because of the properties of the cloning map $\widetilde{\Lambda}$ (see Sec. \ref{sec:problem}) all possible convex mixtures of the partitions produce a correct quantum cloner, i.e. a trace preserving completely positive map.
\end{remark}

\begin{figure}[H]
\begin{center}$
\begin{array}{c}
\includegraphics[scale=0.6]{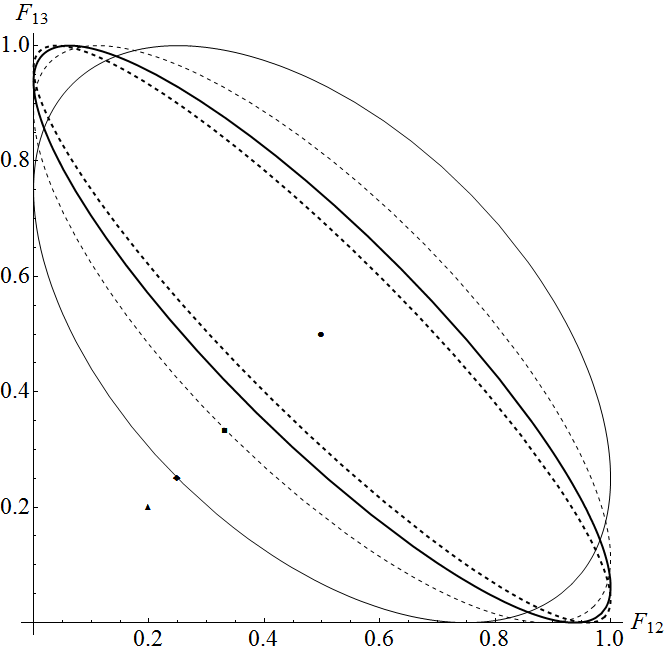}
\end{array}$
\end{center}
\caption{The plot of allowed regions of fidelities for $1\rightarrow 2$ $UQCM$. Views for various dimensions $d$ of the Hilbert space are presented: thin grey line and black point ($d=2$); thin, dashed grey line and square ($d=3$); thick line and diamond ($d=4$); thick, dashed line and triangle ($d=5$). One can see that for $d \rightarrow \infty$ the ellipse is squeezed to the line $F_{13}=-F_{12}+1$ and coordinates of the point obtained from the part $\mathcal{N}$ go to zero.}
\label{fig:ch2}
\end{figure}

\begin{figure}[H]
\begin{center}
$
\begin{array}{ccc}
\includegraphics[width=0.4\textwidth, height=0.35\textheight]{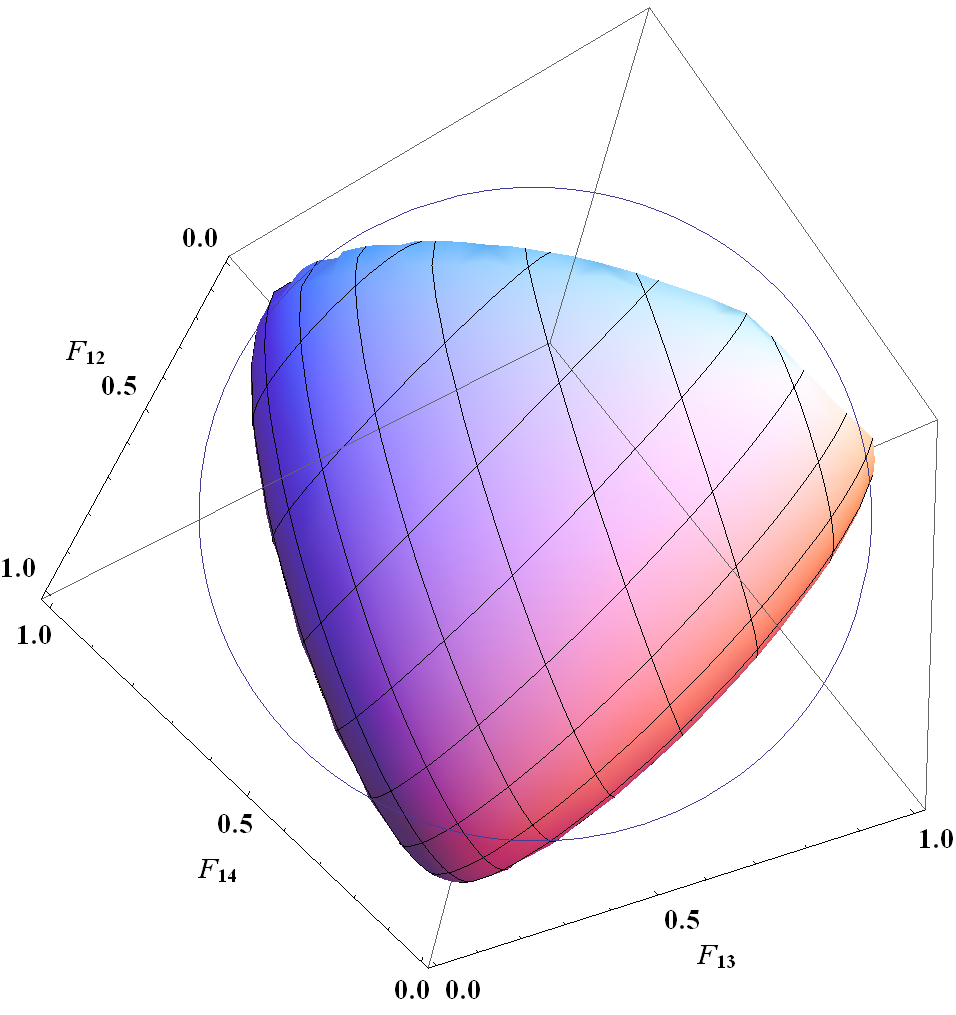} & \qquad
\includegraphics[width=0.4\textwidth, height=0.35\textheight]{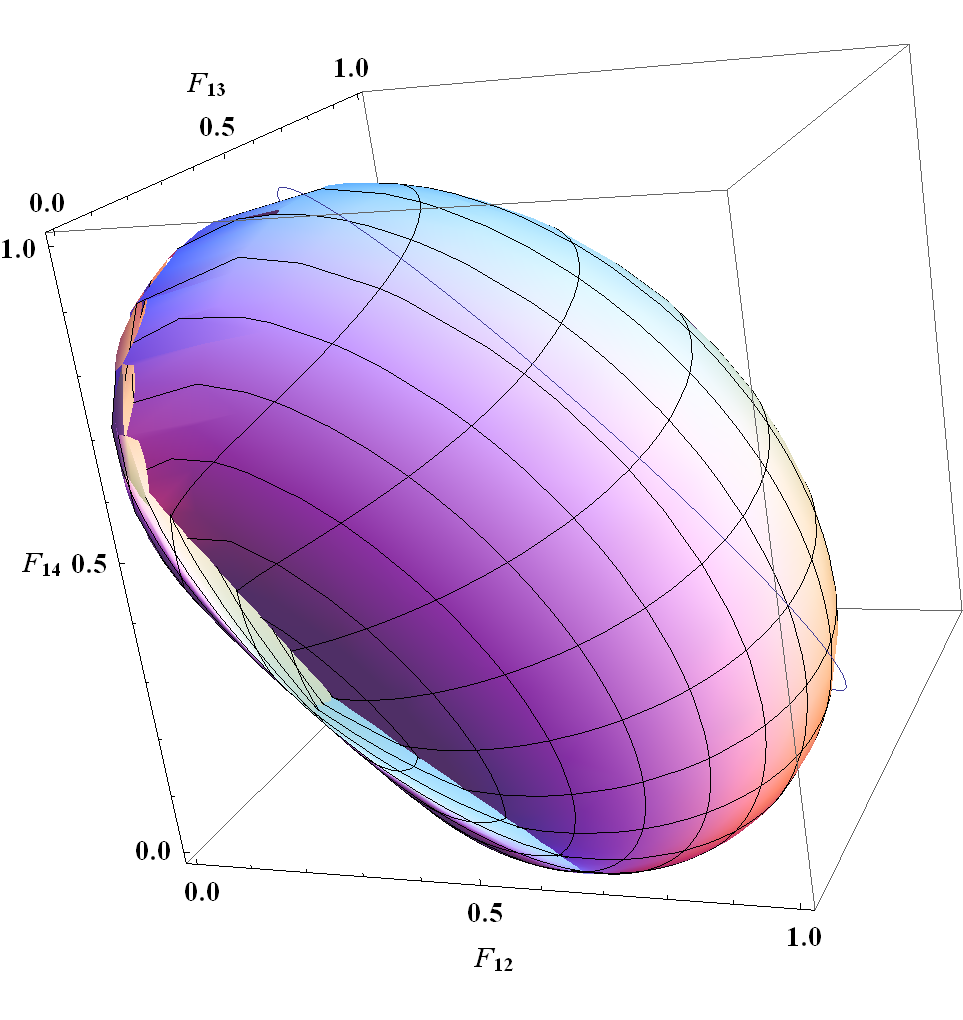}\\
\includegraphics[width=0.4\textwidth, height=0.35\textheight]{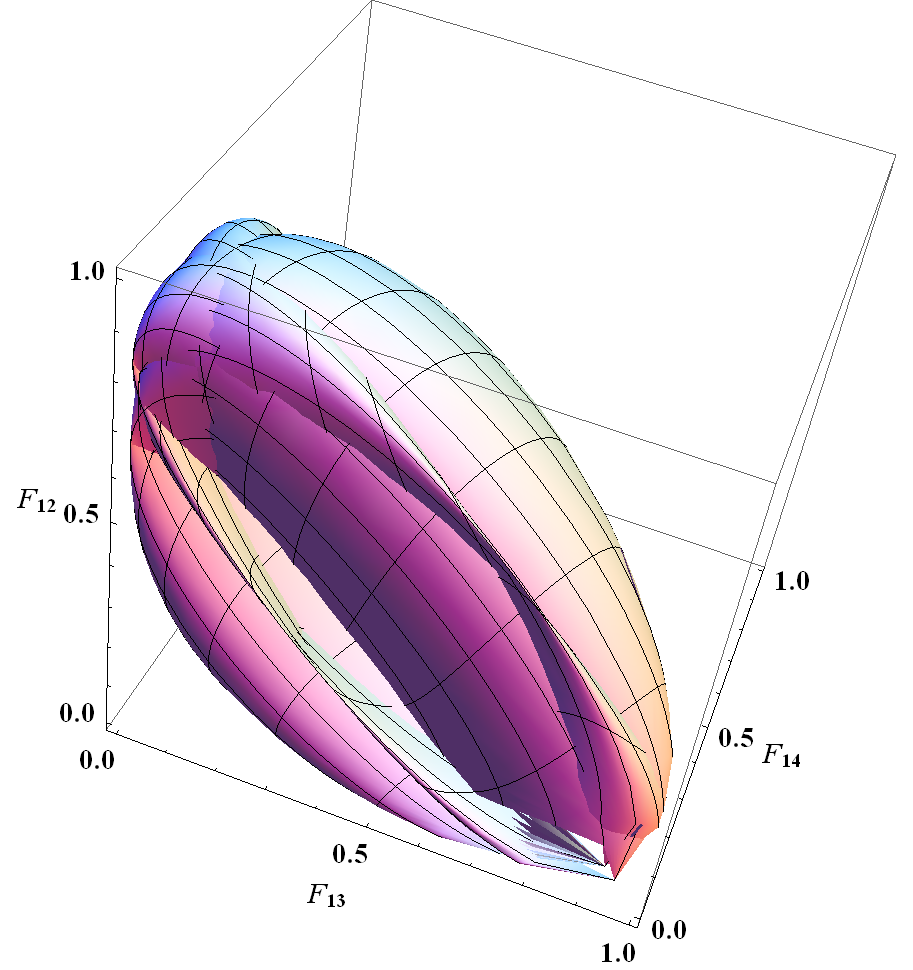} & \qquad
\includegraphics[width=0.4\textwidth, height=0.35\textheight]{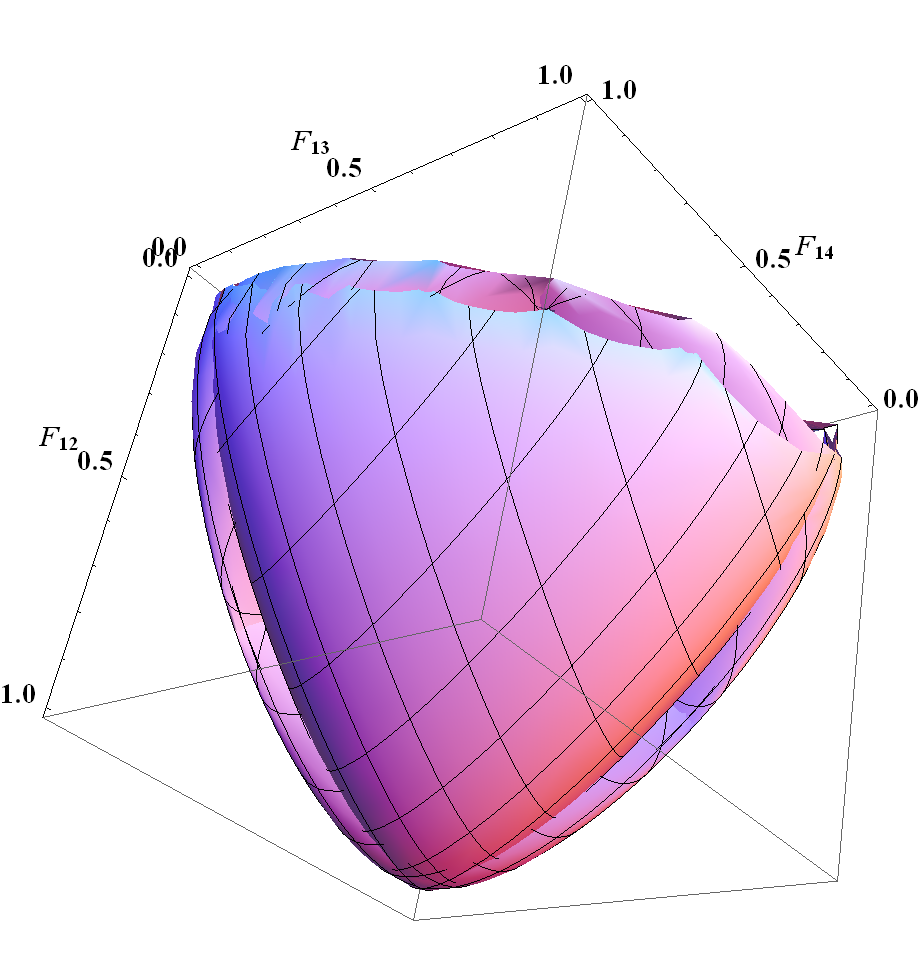}\\
\includegraphics[width=0.4\textwidth, height=0.35\textheight]{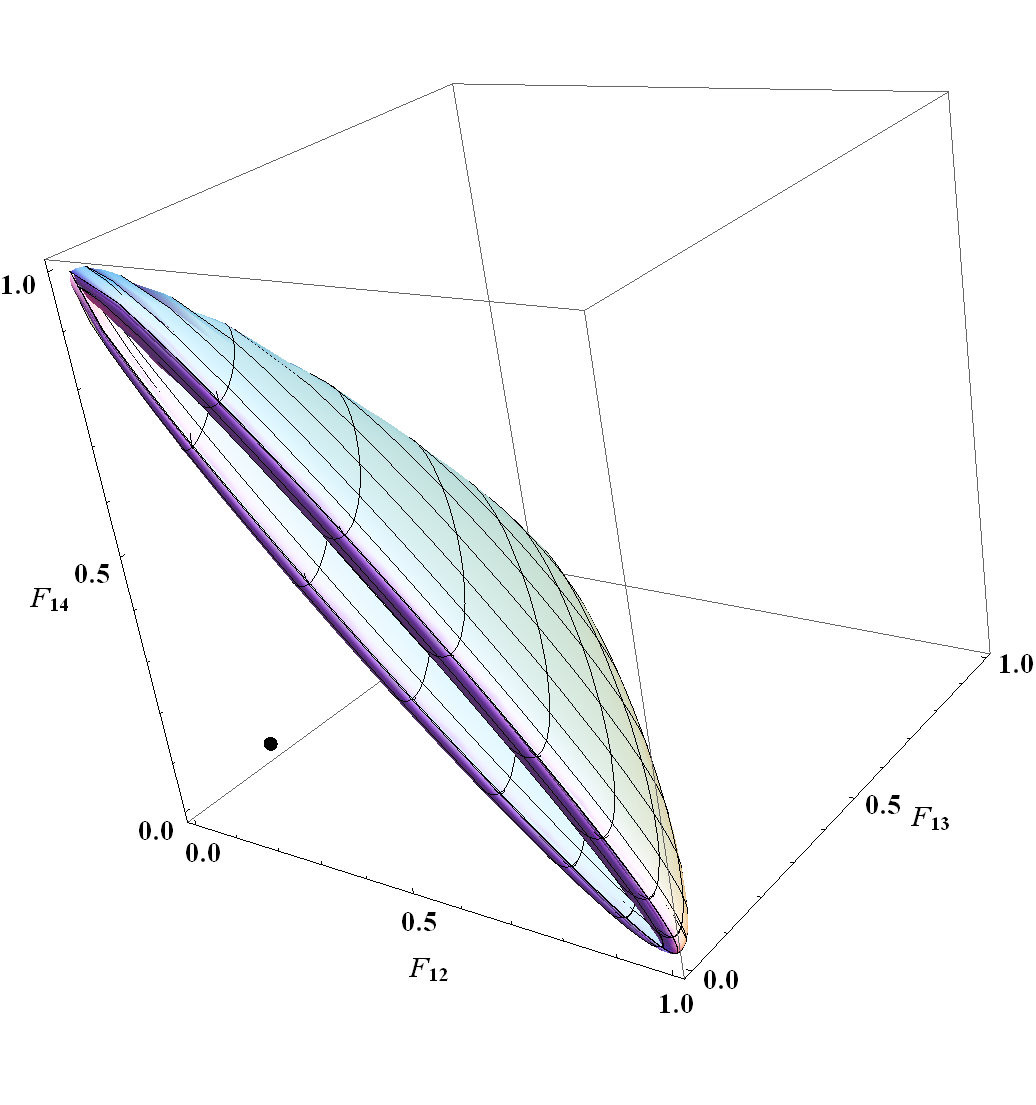} & \qquad
\includegraphics[width=0.4\textwidth, height=0.35\textheight]{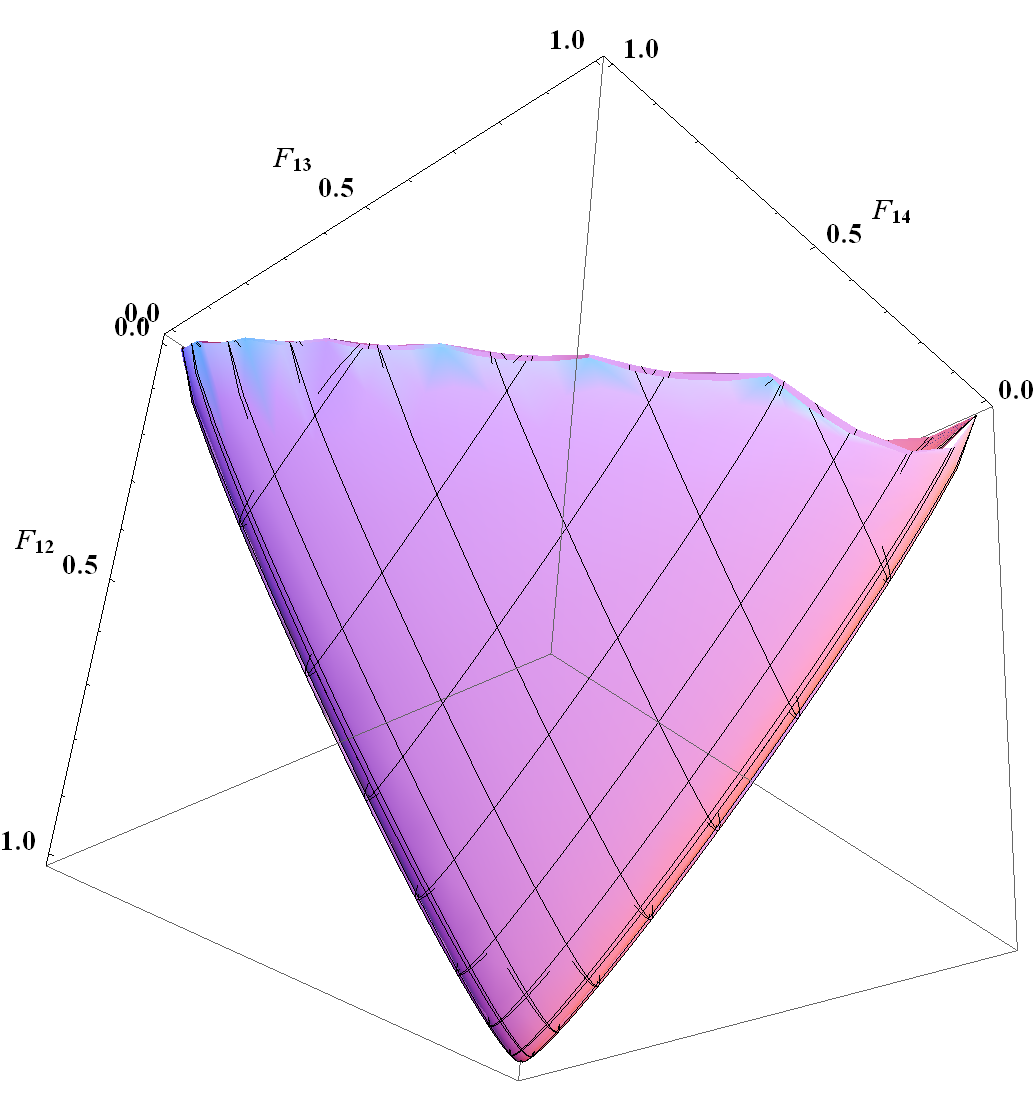}
\end{array}
$
\end{center}
\caption{\textbf{(Color online)} The plot of allowed regions of fidelities for $1\rightarrow 3$ $UQCM$. Views for various dimensions $d$ of the Hilbert space and all allowed irreps are presented. From the top: $d=2$, $d=3$ and $d=10$. One can see that for $d=2$ we match results from~\cite{Cwiklinski2012-cloning} and this is the only case where irreps from $\mathcal{M}$ are two dimensional (in this case we have an ellipse). For higher dimensions $d \rightarrow \infty$ all regions obtained from the part $\mathcal{M}$ are squeezed and coordinates of points from the part $\mathcal{N}$ go to zero.}
\label{fig:ch}
\end{figure}

\begin{figure}[H]
\begin{center}$
\begin{array}{ccc}
\includegraphics[width=0.4\textwidth, height=0.35\textheight]{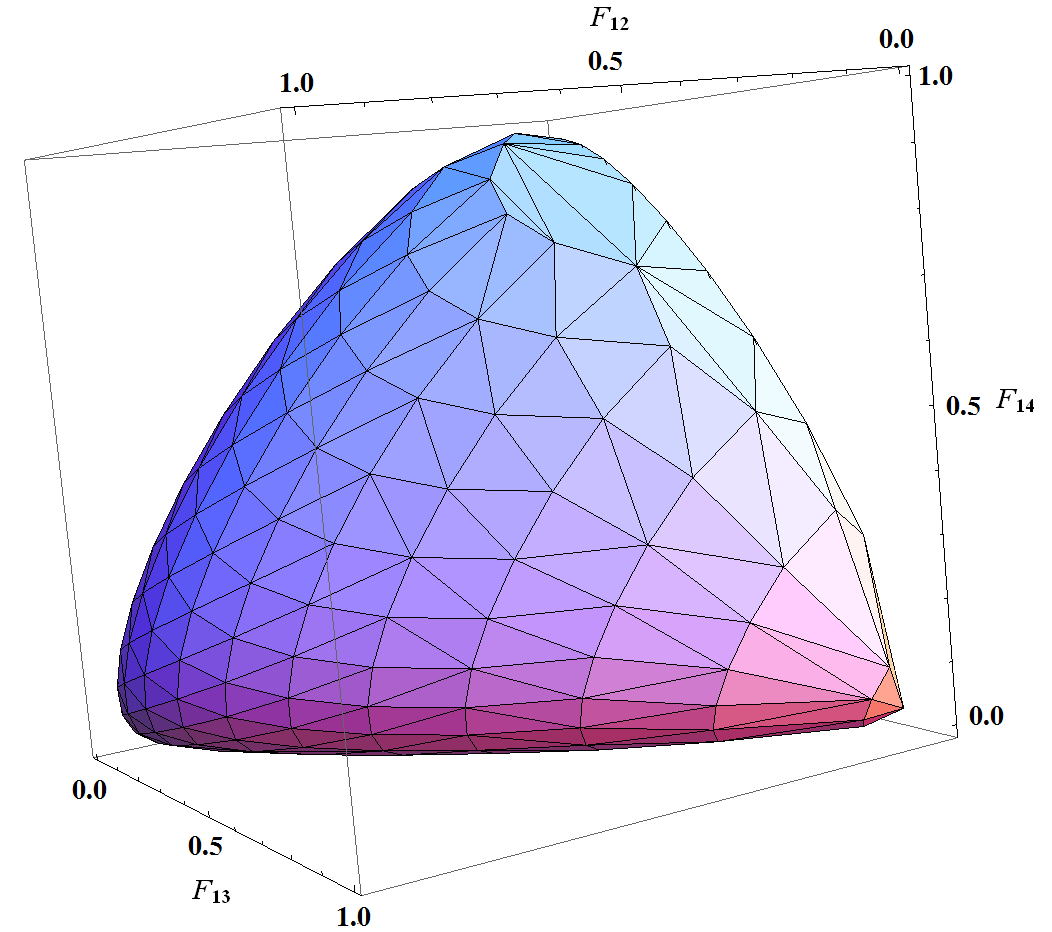} & \quad & \includegraphics[width=0.4\textwidth, height=0.35\textheight]{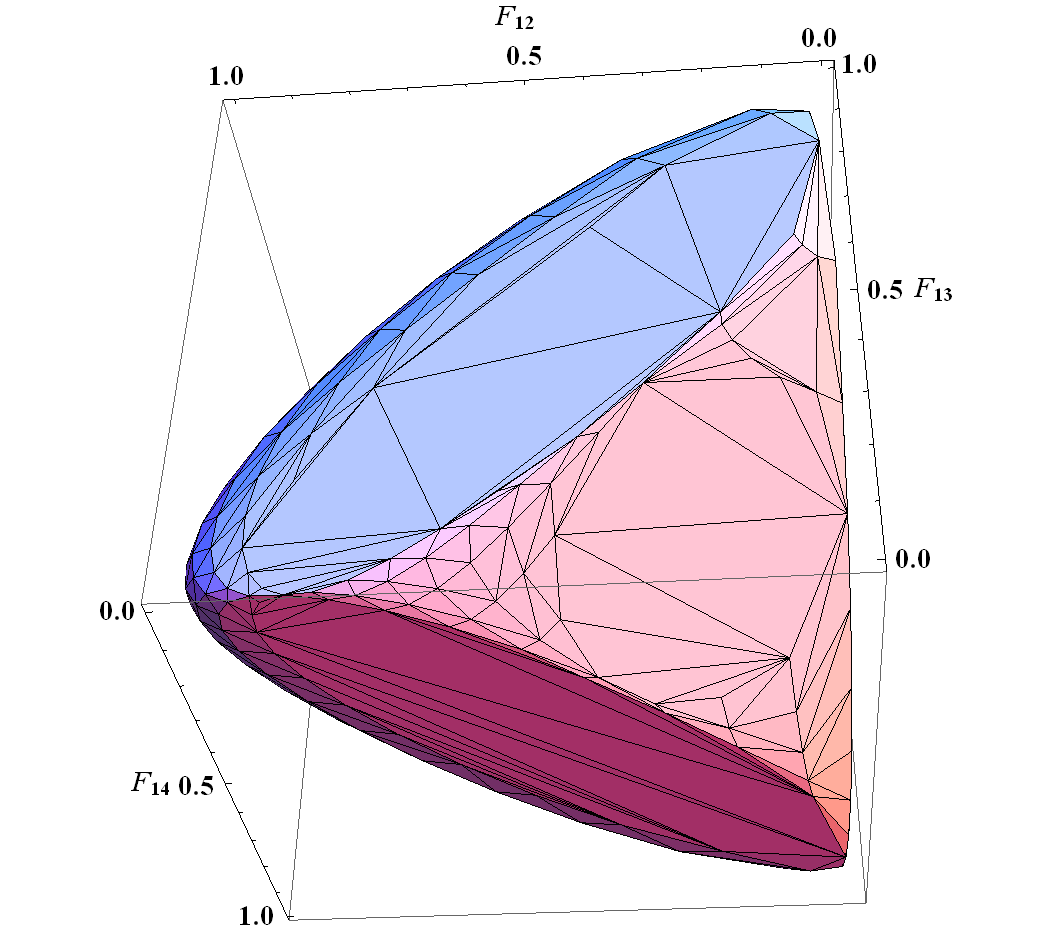}
\end{array}$
\end{center}
\caption{\textbf{(Color online)} As an example convex hull for $1\rightarrow 3$ of $UQCM$ and $d=3$ is presented. }
\label{fig:ch1}
\end{figure}

\subsection{Comparison with other methods}
\label{sec:com}
First of all, let us notice that our method gives correct results (according to the Werner's formula \cite{Werner-cloning1998}) in the case of symmetric cloning (see, \cite{Cwiklinski2012-cloning}, for a possible technique of checking that). What is more, the regions of fidelities obtained for $d=2$ (qubits) match those obtained using Schur-Weyl duality \cite{Cwiklinski2012-cloning}. Last, but not least, our method seems to correctly reproduce results obtained in \cite{Kaszlikowski2012-cloning}, where the solution to the $1 \rightarrow N$ universal asymmetric qudit cloning problem for which the exact
trade-off in the fidelities of the clones for every $N$ and $d$ has been derived. The authors obtained their result using various tools, like the Choi-Jamio{\l}kowski isomorphism \cite{Choi1972, Jamiolkowski} and some variance of the Lieb-Mattis theorem \cite{Lieb1962-LB, Nachtergaele2006-LM}. The crucial part of their proof is the observation that the cloning problem can be mapped to some Heisenberg Hamiltonian on a star. Comparing their technique with ours, one can observe that they solve the problem from the side of the cloning map $\widetilde{\Lambda}$, when we attack it from the side of the $n$-parties quantum state (see, Eq. \eqref{map} and \eqref{map1}).

\section{Conclusions}
We have shown that using a more general version of Schur - Weyl duality, action of the universal $1 \rightarrow N$ quantum cloning machine can be described, allowing to obtain the admissible general region for fidelities. Contrary to other known methods, in our, quantum states are figures of merit.  The method exploits decomposition of (usually big) Hilbert space into blocks of smaller dimensions which, of course, are
easier to deal with.  Fidelity expressions are then quite easy to obtain, one only needs to know representations of all possible irreps for a
given case. Another advantage is that one can consider real pure states in each of the block only when generating convex hulls to obtain an allowed region for fidelities. Let us also notice that it is the first physical application of techniques developed in \cite{Studzinski2013} and, up to our best knowledge, the first graphical presentation of allowed regions for $1 \rightarrow 2$ and $1 \rightarrow 3$ cloners for $d>2$.

%It is worth to say a few words about some advantages of our model 
Let us now shortly discuss the resutls. First of all, suppose that we choose some point that lays outside of the allowed convex hull. Then there does not exist a quantum state that would correspond to that point. On the other hand, whenever we choose points from the convex hull (from inside or from the edge) we are able to derive a family of quantum states for which fidelities are fixed and have values determined by the chosen point. apart from the above-mentioned reconstruction of states from the convex hull, we can try to find, for example, all allowed quantum states which satisfy some required condition for relations between fidelities $F_{1k}$. For example, for $1 \rightarrow 3$ universal cloning machines we can demand the following constraint
\begin{equation}
 F_{12}+F_{13}=2F_{14},
\end{equation}
where we take maximization over $F_{12}$. Such a reconstruction was presented in our previous paper regarding admissible region of fidelities for the qubit case~\cite{Cwiklinski2012-cloning}. Finally, having these states we can reconstruct a cloning machine which returns clones with  fidelities $f_i$, corresponding to fidelities $F_{1i}$ given by the chosen point.

We have also interesting interpretation of the most bottom part of our plots as optimal anti-clones. First of all one can notice that our convex hulls are invariant with respect to rotations around straight line $F_{12}=F_{13}=F_{14}$ by the angle $\beta=2\pi/3$ in the case $1 \rightarrow 3$ $UQCM$ and they are symmetric with respect to the straight line $F_{12}=F_{13}$ in the case $1 \rightarrow 2$ $UQCM$. The most bottom point is determined by the intersection between symmetry line and convex hull and it corresponds to a minimum value of fidelities which are equal in these cases.

In the future, it would be interesting how to obtain optimal clones starting from our method. Numerically, it is not that hard, one just needs to add a cut to the general region to end with optimal region of fidelities. Analytically the answer does not seem to be so trivial, but we still hope that the employed group theoretical techniques are interesting and may provide some new insight into the inner structure of the optimal universal asymmetric quantum cloners.

Finally, let us note that to solve a $M \rightarrow N$ ($M < N$, $M+N=n$) cloning problem, one needs to posses a knowledge of the commutant structure of a $U^{\ot N}\ot \left(U^*\right)^{\ot M}$ transformation, where one has $M$ conjugate elements $U^*$ and $N$ elements $U$ \cite{Studzinski2013,mozrzymas2013}.

{\bf Acknowledgments:} Helpful discussions with Ravishankar Ramanathan are acknowledged. We would also like to thank the anonymous referee for very useful comments regarding our manuscript. M. S. is supported by the International PhD Project "Physics of future quantum-based information technologies": grant MPD/2009-3/4 from Foundation for Polish Science. M. S. and P. \'C. are supported also by grant 2012/07/N/ST2/02873 by National Science Center. M. M. and M.H. are supported by Polish Ministry of Science and Higher Education Grant no. IdP2011 000361. Part of this work was done in National Quantum Information Centre of Gda\'nsk.
\section{Appendix}
\label{app}
\subsection{Algebra of partially transposed permutation operators}
\label{remind}
Here we present short summary of paper~\cite{mozrzymas2013} which is crucial for construction of our results. For the reader convenience we keep here original notation.
It appears that the structure irreducible representations of the algebra $%
\mathcal{A}'_n(d)$ is closely related to the structure of the representation $%
\ind_{S(n-2)}^{S(n-1)}(\varphi ^{\alpha })$ of the group $S(n-1)$ induced by
irreducible representations $\varphi ^{\alpha }$ of the group $S(n-2)$ and
the properties of irreducible representations of $\mathcal{A}'_n(d)$ depends
strongly on the relation between $d$ and $n$. Before presenting the main
ideas of this appendix we have to describe briefly some object appearing in
the structure of the algebra $\mathcal{A}'_n(d)$, in particular the properties
of the induced representation $\ind_{S(n-2)}^{S(n-1)}(\varphi ^{\alpha }).$
The irreducible representations of the group $S(n-2)$ are characterized by the partitions $%
\alpha =(\alpha _{1},...,\alpha _{k})$ of $n-2,$ which describe also the
corresponding Young diagram $Y(\alpha ).$ \ The representation $%
\ind_{S(n-2)}^{S(n-1)}(\varphi ^{\alpha })$ is completely and simply
reducible i.e. we have~\cite{Fulton1991-book-rep}.

\begin{proposition}
\label{prop27}
\be
\ind_{S(n-2)}^{S(n-1)}(\varphi ^{\alpha })=\bigoplus _{\nu }\psi ^{\nu },
\ee
where the sum is over all partitions $\nu =(\nu _{1},...,\nu _{k})$ of $n-1,$
such that their Young diagrams $Y(\nu )$ are obtained from $Y(\alpha )$ by
adding, in a proper way, one box.
\end{proposition}

\begin{definition}~\cite{Curtis}
\label{def8}
Let $\varphi :H\rightarrow M(n,%
%TCIMACRO{\U{2102} }%
%BeginExpansion
\mathbb{C}
%EndExpansion
)$ be a matrix representation of a subgroup $H$ of the group $G.$ \ Then the
matrix form of the induced representation $\pi =\ind_{H}^{G}(\varphi )$ of a
group $G$ induced by an irrep. $\varphi $ of the subgroup $H\subset G$ has
the following block matrix form
\[
\forall g\in G\quad \pi _{ai}^{bj}(g)=(\widehat{\varphi }%
_{ij}(g_{a}^{-1}gg_{b})),
\]%
where $g_{a},$ $a=1,...,[G:H]$ are representatives of the left cosets $G/H$
and
\[
\widehat{\varphi }_{ij}(g_{a}^{-1}gg_{b}) = \left\{ \begin{array}{ll}
\varphi _{ij}(g_{a}^{-1}gg_{b}) & \textrm{if $ \ g_{a}^{-1}gg_{b}\in H$,}\\
0 & \textrm{if $ \ g_{a}^{-1}gg_{b}\notin H$}.
\end{array} \right.
\]
\end{definition}

Before main considerations for the appendix let us introduce some notation.
\begin{notation}
\label{not9}
Any permutation $\sigma \in S(n)$ defines, in a natural and unique way, two
natural numbers $a,b\in \{1,2,...,n\}$%
\[
n=\sigma (a),\qquad b=\sigma (n)
\]

Thus we may characterize any permutation by these two numbers in the
following way%
\[
\sigma \equiv \sigma _{(a,b)}\equiv \sigma_{ab}.
\]%
Note that in general $a,b$ may be different except the case, when one of
them is equal to $n,$ because in this case we have
\[
a=n\Leftrightarrow b=n.
\]%
When $a=n=b,$ then $\sigma (n)=n$ and we will use abbreviation $\sigma
=\sigma _{(n,n)}\equiv \sigma _{n}\in S(n-1) \subset S(n).$
\end{notation}

From Proposition~\ref{prop27} and Definition~\ref{def8} it follows that the induced representation $%
\ind_{S(n-2)}^{S(n-1)}(\varphi ^{\alpha })$ may be described in two bases.
The first one, is the basis of the matrix form of the induced representation
 of the form%
\be
\{e_{i}^{a}(\alpha ):a=1,...,n-1,\quad i=1,...,\dim \varphi ^{\alpha }\},
\ee
where the index $a=1,...,n-1$ describes the the cosets $S(n-1)/S(n-2)$ and
the the index $i=1,...,\dim \varphi ^{\alpha }$ is the index of a matrix
form of $\varphi ^{\alpha }.$ The second one is a basis of the reduced form
of $\ind_{S(n-2)}^{S(n-1)}(\varphi ^{\alpha })$, which is of the
form%
\be
\left\{f_{j_{\nu }}^{\nu }:\psi ^{\nu }\in \ind_{S(n-2)}^{S(n-1)}(\varphi ^{\alpha
}),\quad j_{\nu }=1,...,\dim \psi ^{\nu }\right\}.
\ee

\bigskip The next important objects are the following matrices

\begin{definition}
\label{def28}
For any irreducible representation $\varphi ^{\alpha }$ of the group $S(n-2)$
we define the block matrix
\be
Q_{n-1}^{d}(\alpha )\equiv Q(\alpha )=(d^{\delta _{ab}}\varphi _{ij}^{\alpha
}[(an-1)(ab)(bn-1)])=(Q_{ij}^{ab}(\alpha )),
\ee
where $a,b=1,...,n-1,\quad i,j=1,...,\dim \varphi ^{\alpha }$ and the blocs
of the matrix $Q(\alpha )$ are labeled by indices $(a,b)$ whereas the
elements of the blocks are labeled by the indices of the irreducible
representation $\varphi ^{\alpha }=(\varphi _{ij}^{\alpha })$ of the group $%
S(n-2)$ and $Q(\alpha )\in M((n-1)w^{\alpha },\mathbb{C}).$
\end{definition}

The matrices $Q(\alpha )$ are hermitian and their structure and properties
are described in the~\cite{mozrzymas2013}, where it has been shown, that the
eigenvalues $\lambda _{\nu }$ of the matrix $Q(\alpha )$ are labeled by
the irreducible representations $\psi ^{\nu }\in
\ind_{S(n-2)}^{S(n-1)}(\varphi ^{\alpha })$ and the multiplicity of $\lambda
_{\nu }$ is equal to $\dim \psi ^{\nu }$. The essential for properties of
the irreducible representations of the algebra $\mathcal{A}'_n(d)$ is the
fact, that at most one (up to the multiplicity) eigenvalue $\lambda _{\nu }$
of the matrix $Q(\alpha )$ may be equal to zero~\cite{mozrzymas2013,Studzinski2013}.

The structure of the algebra $\mathcal{A}'_n(d)$ is the following

\begin{theorem}
\label{th29}
The algebra $\mathcal{A}'_n(d)$ is a direct sum of two ideals
\be
\mathcal{A}'_n(d)=\mathcal{M}\oplus \mathcal{N}
\ee
and the ideals $\mathcal{M}$ and $\mathcal{N}$ has different structures.
\begin{enumerate}[a)]
\item The ideal $\mathcal{M}$ is of the form
\be
\mathcal{M}=\bigoplus _{\alpha }U(\alpha ),
\ee
where $U(\alpha )$ are ideals of the algebra $\mathcal{A}'_n(d)$
characterized by the irreducible representations $\varphi ^{\alpha }$ of the
group $S(n-2)$, such that $\varphi ^{\alpha }\in \opV_{d}[S(n-2)]$ and
\be
U(\alpha )=\Span_{\mathbb{C}}\{u_{ij}^{ab}(\alpha ):a,b=1,...,n-1,\quad i,j=1,...,w^{\alpha }\}
\ee
with
\be
u_{ij}^{ab}(\alpha )u_{kl}^{pq}(\beta )=\delta _{\alpha \beta
}Q_{ik}^{bp}(\alpha )u_{il}^{aq}(\alpha ).
\ee
The ideals $U(\alpha )$ are matrix ideals such that
\be
U(\alpha )\simeq M(\rank Q(\alpha ),\mathbb{C}),
\ee
in particular when $\det Q(\alpha )\neq 0$ we have
\be
U(\alpha )\simeq M((n-1)\dim \varphi ^{\alpha },\mathbb{C}).
\ee
\item The ideal $\mathcal{N}$ has the following structure
\be
\mathcal{N}\simeq \bigoplus _{\nu }M(\dim \psi ^{\nu },\mathbb{C}),
\ee
where the matrix ideals $M(\dim \psi ^{\nu },\mathbb{C})$ are generated by irreducible representations $\psi ^{\nu }$ of the group
$S(n-1)$ that are included in the representation $\opV_{d}[S(n-1)]$ i.e. $\psi
^{\nu }$ are such that $d\geq h(\nu )$.
\end{enumerate}
\end{theorem}

\bigskip The matrix ideals contained in the ideals $\mathcal{M}$ and $\mathcal{N}$ contains all
minimal left ideals i.e. all irreducible representations of the algebra $%
\mathcal{A}'_n(d)$. The next theorems describes all these representations.

The structure of the irreducible representations of the algebra $%
\mathcal{A}'_n(d)$, included in the ideal $\mathcal{M}$, is completely determined by
irreducible representations $\varphi ^{\alpha }$ of the group $S(n-2)$,
therefore we will denote them $\Phi _{A}^{\alpha }.$

\begin{theorem}
\label{th30}
The irreducible representations $\Phi _{A}^{\alpha }$ of the algebra $%
\mathcal{A}'_n(d)$ contained in the ideal $U(\alpha )\subset \mathcal{M}$ (see Theorem~\ref{th29}) are indexed
by the irreducible representations $\varphi ^{\alpha }$ of the group $S(n-2)$%
, such that $\varphi ^{\alpha }\in \opV_{d}[S(n-2)]$ and if $\{f_{j_{\nu
}}^{\nu }:\psi ^{\nu }\in \ind_{S(n-2)}^{S(n-1)}(\varphi ^{\alpha }),\quad
j_{\nu }=1,...,\dim \psi ^{\nu }\}$ is the reduced basis of the induced
representation $\ind_{S(n-2)}^{S(n-1)}(\varphi ^{\alpha })$, then the vectors
$\{f_{j_{\nu }}^{\nu }:\lambda _{\nu }\neq 0\}$ from the basis of the
irreducible representation of the algebra $\mathcal{A}'_n(d)$ and the natural
generators of $\mathcal{A}'_n(d)$ act on it in the following way%
\be
\opV'(an){}f_{j_{\nu}}^{\nu}(\alpha )=\sum_{\rho ,j_{\rho
}}\sum_{k}\sqrt{\lambda _{\rho }}z^{\dagger}(\alpha )_{j_{\rho }k}^{\rho
a}z(\alpha )_{kj_{\nu }}^{a\nu }\sqrt{\lambda _{\nu }}f_{j_{\rho }}^{\rho}(\alpha ),
\ee
where the summation is over $\rho $ such that $\lambda _{\rho }\neq 0.$
Due to the condition $\varphi ^{\alpha }\in V_{d}[S(n-2)]$
the eigenvalues $\lambda _{\nu }$ of $\ Q(\alpha )$ are
non-negative. The unitary matrix $Z(\alpha )=(z(\alpha )_{kj_{\nu }}^{a\nu
})$ has the form
\be
z(\alpha )_{kj_{\nu }}^{a\nu }=\frac{\dim \psi ^{\nu }}{\sqrt{%
N_{j_{\nu }}^{\nu }}(n-1)!}\sum_{\sigma \in S(n-1)}\psi _{j_{\nu }j_{\nu
}}^{\nu }(\sigma ^{-1})\delta _{a\sigma (q)}\varphi _{kr}^{\alpha
}[(an-1)\sigma (qn-1)],
\ee
with
\be
N_{j_{\nu }}^{\nu }=\frac{\dim \psi ^{\nu }}{(n-1)!}%
\sum_{\sigma \in S(n-1)}\psi _{j_{\nu }j_{\nu }}^{\nu }(\sigma ^{-1})\delta
_{q\sigma (q)}\varphi _{rr}^{\alpha }[(qn-1)\sigma (qn-1)],
\ee
where the indices $q=1,..,n-1,$\textbf{\ }$r=1,..,\dim \varphi
^{\alpha }$ are fixed and such that $N_{j_{\nu }}^{\nu }>0$. For more details see~\cite{mozrzymas2013}. Whenever $\sigma _{n}\in S(n-1)$ we have
\be
V(\sigma _{n})f_{j_{\nu }}^{\nu }(\alpha )=\sum_{\rho ,j_{\rho }}\psi
_{i_{\nu }j_{\nu }}^{\nu }(\sigma _{n})f_{i_{\nu }}^{\nu }(\alpha ).
\ee
In particular when $\det Q(\alpha )\neq 0,$ (i.e. when all $\lambda _{\nu
}\neq 0$) then the representation $\Phi _{A}^{\alpha }$ is the induced
representation $\ind_{S(n-2)}^{S(n-1)}(\varphi ^{\alpha })$ (in the reduced
form) for the subalgebra $\opV_{d}[S(n-1)]\subset $ $\mathcal{A}'_n(d)$. In this
case the dimension of the irreducible representation is equal to
\be
\dim \Phi _{A}^{\alpha }=(n-1)\dim \varphi ^{\alpha }=\dim
(\ind_{S(n-2)}^{S(n-1)}(\varphi ^{\alpha })).
\ee
When $\det Q(\alpha )=0,$ (i.e. when one, up to the multiplicity, eigenvalue
$\lambda _{\theta }$ of $Q(\alpha )$ is equal to $0$)$,$ then the
irreducible representation of $\mathcal{A}'_n(d)$ is defined on a subspace $%
\{y_{j_{\nu }}^{\nu }:\lambda _{\nu }\neq \lambda _{\theta }\}$ of the
representation space $\ind_{S(n-2)}^{S(n-1)}(\varphi ^{\alpha })$ and the
representation has dimension is equal to
\be
\dim \Phi _{A}^{\alpha }=\dim (_{S(n-2)}^{S(n-1)}(\varphi ^{\alpha
}))-\dim \psi ^{\theta }=\rank Q(\alpha ).
\ee
This case takes the place when
\be
d=i-\alpha _{i}-1
\ee
for some $\alpha _{i}$ in the partition $\alpha =(\alpha _{1},..,\alpha
_{i},..,\alpha _{k})$ characterizing the irreducible representation $\varphi
^{\alpha }$, under condition that $\nu =(\alpha _{1},..,\alpha
_{i}+1,..,\alpha _{k})$ characterizes the representation $\psi ^{\nu }$ of $
S(n-1)$.

The ideal $U(\alpha )$ is a direct sum of $\dim \Phi _{A}^{\alpha }$ of
irreducible representations $\Phi _{A}^{\alpha }.$
\end{theorem}
In particular matrices $z(\alpha )_{kj_{\nu }}^{a\nu }$ diagonalize matrix $Q(\alpha
)_{kl}^{ab}$, i.e. we have following
\begin{proposition}
\be
\label{A2}
\sum_{ak}\sum_{bl}z^{\dagger}(\alpha )_{j_{\rho }k}^{\rho a}Q(\alpha
)_{kl}^{ab}z(\alpha )_{lj_{\mu }}^{b\mu }=\delta ^{\rho \mu }\delta
_{j_{\rho }j_{\mu }}\lambda _{\mu }
\ee
and the columns of the matrix $Z(\alpha )=(z(\alpha )_{kj_{\nu }}^{a\nu })$
are eigenvectors of the matrix $Q(\alpha ).$
\end{proposition}
The formula for the eigenvalues $\lambda _{\nu }$ of matrices $Q(\alpha )$
is derived in the~\cite{mozrzymas2013}.

\begin{remark}
\label{rem31}
\bigskip Note that even if $\dim \varphi ^{\alpha }=1$, we have $\dim \Phi
^{\alpha }=n-1.$
\end{remark}

The matrix forms of these representations are the
following

\begin{proposition}
\label{prop32}
\bigskip In the reduced matrix basis $\{f_{j\nu }^{\nu }:\nu \neq \theta \}$
of the ideal $U(\alpha )$ the natural generators $\opV(\sigma _{ab})^{t_n }$
and $\opV(\sigma _{n})$ of $\mathcal{A}'_n(d)$ are represented by the following
matrices
\be
[\mbV'_{\alpha}(an)]_{j_{\rho }j_{\nu }}^{\rho \nu
}=\sum_{k=1,..,\dim \varphi ^{\alpha }}\sqrt{\lambda _{\rho }}z^{\dagger}(\alpha
)_{j_{\rho }k}^{\rho a}z(\alpha )_{kj_{\nu }}^{a\nu }\sqrt{\lambda _{\nu }}
:\rho ,\nu \neq \theta ,
\ee
\be
[\mbV_{\alpha}(\sigma _{n})]_{j_{\nu ^{\prime }}j_{\nu }}^{\nu ^{\prime
}\nu }=\delta ^{\nu ^{\prime }\nu }\psi _{j_{\nu ^{\prime }}j_{\nu }}^{\nu
}(\sigma _{n}).
\ee
\end{proposition}

From the properties of the matrix $Q(\alpha )$ (\cite{mozrzymas2013}) one gets

\begin{proposition}
\label{prop33}
If $d>n-2$, then $\det Q(\alpha )\neq 0$ and the irreducible representations
$\Phi _{A}^{\alpha }$ described in Th.~\ref{th30} are induced representation $
\ind_{S(n-2)}^{S(n-1)}(\varphi ^{\alpha })$ for the subalgebra $
\opV_{d}[S(n-1)]\subset $ $\mathcal{A}'_n(d)$, so their dimension is equal to $
(n-1)\dim \varphi ^{\alpha }.$ When $d\leq n-2$, then for some $\varphi
^{\alpha }$ it may appear that $\det Q(\alpha )=0$ and consequently the
irreducible representation $\Phi ^{\alpha }$ of $\mathcal{A}'_n(d)$ is define
on a subspace of the irreducible representation $\ind_{S(n-2)}^{S(n-1)}(
\varphi ^{\alpha })$.
\end{proposition}

% When $\det Q(\alpha )\neq 0$, the equivalent form of the irreducible
% representation $\Phi _{A}^{\alpha }$ form Th.~\ref{th30}, in the basis $
% \{e_{i}^{a}(\alpha ):a=1,...,n-1,\quad i=1,...,\dim \varphi ^{\alpha }\}$ is
% given in the Prop.~\ref{prop48},~\ref{prop50}.

The representations of the algebra $\mathcal{A}'_n(d)$ included in the ideal $
\mathcal{N}$ are much simpler.

\begin{theorem}
\label{th34}
Each irreducible representation $\psi ^{\nu }$ of the group $S(n-1)$, which
appears in the decomposition of the ideal $\mathcal{N}$ given in the Th.~\ref{th30} b), (i.e. $
\psi ^{\nu }$ $\in \opV_{d}[S(n-1)]\Leftrightarrow d\geq h(\nu ))$ defines
irreducible representations $\Psi ^{\nu }$ of the algebra $\mathcal{A}'_n(d)$
in the following way
\be
\Psi ^{\nu }(a)= \left\{ \begin{array}{ll}
0 & \textrm{if $ \ a\in \mathcal{M}$,}\\
\psi ^{\nu }(\sigma _{n}) & \textrm{if $ \ a=\sigma _{n}\in S(n-1)$.}
\end{array} \right.
\ee
So in this representation the non-invertible element of the ideal $\mathcal{M}$ are
represented trivially by zero and therefore we call these representation of
the algebra $\mathcal{A}'_n(d)$ semi-trivial. The matrix forms of these
representations \ are simply matrix forms of the irreducible representations
of the group algebra $\mathbb{C}\lbrack S(n-1)]\subset A_{n}'(d)$ and zero matrices for the elements
of the ideal $\mathcal{M}$.
\end{theorem}

\begin{corollary}
\label{col35}
\bigskip All irreducible representations of the algebra $\mathcal{A}'_n(d)$
of dimension one are included in the ideal $\mathcal{N}$. In particular, because the
irreducible identity representation $\psi ^{\id}$ of $S(n-1)$ is always
contained in $\opV_{d}[S(n-1)]$, the algebra $\mathcal{A}'_n(d)$ has a trivial
representation, in which the elements of the ideal $M$ are represented by
zero and the elements $\opV_{d}(\sigma ):\sigma \in S(n-1)$ are represented by
number $1$.
\end{corollary}

\subsection{Auxiliary lemmas}
\label{aux}
After short summary of paper~\cite{mozrzymas2013} given in the previous subsection we prove here the crucial lemma which says that matrices $z(\alpha )_{kj_{\nu }}^{a\nu }$ are unitary (real orthogonal) and then we conclude that representation matrices in the reduced matrix basis are hermitian (symmetric). We start from the following proposition:
 \label{Ap2}
 \begin{proposition}
 \label{propaux}
Suppose that all  representations $\psi ^{\nu }$ of $S(n-1)$ and $\varphi
^{\alpha }$ of $S(n-2)$ are unitary (real orthogonal) then the matrix
\be
z(\alpha )_{kj_{\nu }}^{a\nu }=\frac{\dim \psi ^{\nu }}{\sqrt{%
N_{j_{\nu }}^{\nu }}(n-1)!}\sum_{\sigma \in S(n-1)}\psi _{j_{\nu }j_{\nu
}}^{\nu }(\sigma ^{-1})\delta _{a\sigma (q)}\varphi _{kr}^{\alpha
}[(an-1)\sigma (qn-1)],
\ee
where%
\be
N_{j_{\nu }}^{\nu }=\frac{\dim \psi ^{\nu }}{(n-1)!}%
\sum_{\sigma \in S(n-1)}\psi _{j_{\nu }j_{\nu }}^{\nu }(\sigma ^{-1})\delta
_{q\sigma (q)}\varphi _{rr}^{\alpha }[(qn-1)\sigma (qn-1)],
\ee
is unitary (real orthogonal).
\end{proposition}

\begin{proof}
We will prove the orthogonal case, proving that
\be
\sum_{c,k}z(\alpha )_{kj_{\mu }}^{c\mu }z(\alpha )_{kj_{\nu }}^{c\nu
}=\delta ^{\mu \nu }\delta _{j_{\mu }j_{\nu }.}
\ee
Using the definition of the matrix $z(\alpha )$ we get that $\operatorname{LHS}$ of the
above equation is equal to%
\be
\frac{\dim \psi ^{\nu }\dim \psi ^{\mu }}{\sqrt{N_{j_{\nu }}^{\nu }}\sqrt{%
N_{j_{\mu }}^{\mu }}((n-1)!)^{2}}\sum_{\sigma ,\rho \in
S(n-1)}\sum_{c,k}\psi _{j_{\mu }j_{\mu }}^{\mu }(\rho ^{-1})\psi _{j_{\nu
}j_{\nu }}^{\nu }(\sigma ^{-1})\delta _{c\rho (q)}\delta _{c\sigma
(q)}\varphi _{kr}^{\alpha }[(cn-1)\rho (qn-1) ]\varphi _{kr}^{\alpha }[(cn-1)\sigma (qn-1)]=
\ee
\be
\frac{\dim \psi ^{\nu }\dim \psi ^{\mu }}{\sqrt{N_{j_{\nu }}^{\nu }}\sqrt{%
N_{j_{\mu }}^{\mu }}((n-1)!)^{2}}\sum_{\sigma ,\rho \in S(n-1)}\psi _{j_{\mu
}j_{\mu }}^{\mu }(\rho ^{-1})\psi _{j_{\nu }j_{\nu }}^{\nu }(\sigma
^{-1})\delta _{\rho ^{-1}\sigma (q)q}\varphi _{rr}^{\alpha }[\rho ^{-1}\sigma ].
\ee
Substituting $\gamma =\rho ^{-1}\sigma \in S(n-2)\subset S(n-1)$ (which
follows  from $\delta _{\rho ^{-1}\sigma (q)q})$ we get%
\be
\sum_{c,k}z(\alpha )_{kj_{\mu }}^{c\mu }z(\alpha )_{kj_{\nu }}^{c\nu }=\frac{%
\dim \psi ^{\nu }\dim \psi ^{\mu }}{\sqrt{N_{j_{\nu }}^{\nu }}\sqrt{%
N_{j_{\mu }}^{\mu }}((n-1)!)^{2}}\sum_{\rho \in S(n-1)\gamma \in
S(n-2)}\sum_{k_{\nu }}\psi _{j_{\mu }j_{\mu }}^{\mu }(\rho ^{-1})\psi
_{j_{\nu }k_{\nu }}^{\nu }(\rho )\psi _{k_{\nu }j_{\nu }}^{\nu }(\gamma
^{-1})\delta _{\gamma (q)q}\varphi _{rr}^{\alpha
}[\gamma ].
\ee
Now using the orthogonality relations for the irreducible representations $%
\psi ^{\nu }$ of $S(n-1)$ we obtain%
\be
\sum_{c,k}z(\alpha )_{kj_{\mu }}^{c\mu }z(\alpha )_{kj_{\nu }}^{c\nu }=\frac{%
\dim \psi ^{\nu }}{\sqrt{N_{j_{\nu }}^{\nu }}(n-1)!}\sum_{\gamma \in
S(n-2)}\delta ^{\mu \nu }\delta _{j_{\mu }j_{\nu }}\psi _{j_{\nu }j_{\nu
}}^{\nu }(\gamma ^{-1})\delta _{\gamma (q)q}\varphi _{rr}^{\alpha }[\gamma ]=\delta ^{\mu \nu }\delta _{j_{\mu }j_{\nu
}}.
\ee
The proof for the unitary case is similar.
\end{proof}

\begin{corollary}
\label{symm}
Suppose that all representations $\psi ^{\nu }$ of $S(n-1)$ and $\varphi
^{\alpha }$ of $S(n-2)$ are unitary (real orthogonal) then the
representation matrices (in  the reduced matrix basis $\{f_{j\nu }^{\nu
}:\nu \neq \theta \}$ of the ideal $U(\alpha ))$
\be
\label{eee}
[\mbV'_{\alpha}(an)]_{j_{\rho }j_{\nu }}^{\rho \nu
}=\sum_{k=1,..,\dim \varphi ^{\alpha }}\sqrt{\lambda _{\rho }}z^{+}(\alpha
)_{j_{\rho }k}^{\rho a}z(\alpha )_{kj_{\nu }}^{a\nu }\sqrt{\lambda _{\nu }}%
:\rho ,\nu \neq \theta ,
\ee
are hermitian (real symmetric). In the orthogonal case we have replace hermitian conjugation $\dagger$ in the equation~\eqref{eee} by normal transposition $\operatorname{T}$.
\end{corollary}
Indeed unitarity (orthogonality) of matrices $z(\alpha)^{a\nu}_{kj_{\nu}}$ from Proposition~\ref{propaux} allows us to write $z^{+}(\alpha)^{a\nu}_{kj_{\nu}}=z(\alpha)^{\nu a}_{j_{\nu}k}$. Now writing explicitly matrix elements for $[\mbV'_{\alpha}(an)]_{j_{\rho }j_{\nu }}^{\rho \nu
}$ and $[\mbV'_{\alpha}(an)]_{j_{\nu }j_{\rho }}^{\nu \rho}$ together with unitarity (orthogonality) properties from Proposition~\ref{propaux} we obtain statement of Corollary~\ref{symm}.

\subsection{Proofs of the theorems from the main text}
\begin{proof}[Proof of Lemma~\ref{FF}]
From the definition of a fidelity we can write
\be
\label{partialf}
F_{1k}=\<\psi_{1k}|\rho_{1k}|\psi_{1k}\>=\tr\left(\rho_{1k}|\psi_{1k}\>\<\psi_{1k}|\right)=\frac{1}{d}\tr\left(\rho_{1k}\opV'(1k)\right),
\ee
where $\frac{1}{d}\opV'(1k)=|\psi_{1k}\>\<\psi_{1k}|$, $\rho_{1k} = \tr_{\overline{1k}} \rho_{1\ldots n}$ and $\tr_{\overline{1k}}$ denote partial trace over all systems except $1$ and $k$.

Now we can use decomposition of which we mentioned in Eq.~\ref{decomp1} to represent $\opV(1k)$ and $\rho_{1\ldots n}$:
\be
\opV'(1k)=\bigoplus_{\alpha} \operatorname{\mbI}_{r(\alpha)} \ot \mbV'_{\alpha}(1k), \ \ \rho_{1\ldots n}=\bigoplus_{\alpha} \operatorname{\mbI}_{r(\alpha)} \ot \widetilde{\rho}^{\alpha} \label{Vr},
\ee
where $\alpha$ runs over all partitions of $n-2$.
Inserting~\eqref{Vr} into \eqref{partialf}, we have:
\be
\label{res}
\begin{split}
F_{1k}&=\frac{1}{d}\left[\left(\bigoplus_{\mu} \operatorname{\mbI}_{r(\mu)} \ot \widetilde{\rho}^{\mu}\right) \left(\bigoplus_{\alpha} \operatorname{\mbI}_{r(\alpha)} \ot \mbV'_{\alpha}(1k)\right)\right]=\frac{1}{d}\tr\left(\bigoplus_{\alpha} \operatorname{\mbI}_{r(\alpha)} \ot \widetilde{\rho}^{\alpha}\mbV'_{\alpha}(1k) \right)=\\
&=\frac{1}{d}\sum_{\alpha}\tr\left(\rho^{\alpha}\mbV'_{\alpha}(1k)\right)=\frac{1}{d}\sum_{\lambda}\tr\left(\rho^{\alpha}\mbV'_{\alpha}(k-1n)\right),
\end{split}
\ee
where the last equality follows from Eq.~\ref{mapping}.
Now, one can see that Eq.~\eqref{res} can be written as:
\be
\label{comb}
F_{1k}=\sum_{\alpha}F_{1k}^{\alpha},
\ee
where $F_{1k}^{\alpha}=\frac{1}{d}\sum_{\alpha}\tr\left(\rho^{\alpha}\mbV'_{\alpha}(k-1n)\right)$, $\rho^{\alpha} = d_{\alpha} \widetilde{\rho}^{\alpha}$ and $d_{\alpha}$ stands for the dimension of irrep labeled by partition $\alpha$.
\end{proof}

\begin{proof}[Proof of Fact~\ref{fact}]
Reader can prove this fact by direct calculations. Namely, one has to compute fidelity between state which is a product in $1|2\ldots n$ cut and maximally entangled state $|\psi_{1k}\>$:
\be
F_{1k}^{\mathcal{N}}=\frac{1}{d}\<\psi_{1k}|\tr_{\overline{1k}}\left(\text{\noindent\(\mathds{1}\)}_1 \ot \rho_{2\ldots n}\right)|\psi_{1k}\>=\frac{1}{d}\<\psi_{1k}|\text{\noindent\(\mathds{1}\)}_1 \ot \rho_{k}|\psi_{1k}\>=\frac{1}{d}\tr \rho_k=\frac{1}{d}.
\ee
%namely $F^{\mathcal{N}}_{1k}=\tr(\hat{\rho} \Phi^+_{1k})=1/d$, where $\Phi^+_{1k}=|\psi_{1k}\>\<\psi_{1k}|$.
\end{proof}

\begin{proof}[Proof of Theorem~\ref{thm:main}]
The proof is similar to that in~\cite{Cwiklinski2012-cloning}. Only difference is the fact that now the fidelities look like as in Eq. \eqref{ffv}.
\end{proof}

\begin{proof}[Proof of Lemma~\ref{real}]
The proof goes as in~\cite{Cwiklinski2012-cloning}. The only new thing in the proof is that matrices of irreps for transpositions $(in)$, where $1\leq i \leq n-1$ are symmetric (see  Appendix~\ref{aux}, Corollary~\ref{symm}).
\end{proof}

\subsection{Fidelity region for each irreducible space and some applications}
\label{appl}

In this section we provide some technical details regarding construction of admissible region of fidelities for $1 \rightarrow N$ $UQCM$. We focus here for clarity on the case when $N=3$, then we have two non-trivial irreps $\alpha_1=(2)$ and $\alpha_2=(1,1)$. We also restrict here to dimensions $d \geq 3$ to omit discussion about dimension of irrep $\alpha_2$, but of course construction in this situation is the same. For any $d \geq 3$ non-trivial irreps have the same dimension equal to three, thanks to this and Lemma~\ref{real} we can write an arbitrary pure state as $|\psi^{\alpha_i} \> =(
a_1, a_2,a_3)^{\operatorname{T}}$ and corresponding density matrix as
$\rho^{\alpha_i} =\begin{pmatrix}
a_1^2 & a_1a_2 & a_1a_3\\
a_1a_2 & a_2^2 & a_2a_3\\
a_1a_3&  a_2a_3 & a_3^2
\end{pmatrix}$, where $a_1^2+a_2^2+a_3^2=1$ and $i=1,2$. Now putting for example density matrix $\rho^{(2)}$  into equation~\ref{ffv} from Lemma~\ref{FF}, together with irreps $\mbV'_{(2)}(k \ n-1)$ from formula~\eqref{alpha1} we obtain following set of equations:
\be
\label{form1}
\begin{split}
F^{(2)}_{12}&=\frac{1}{18 d}\left(a_1^2 (d-1)-2 \sqrt{3} a_1 a_2 (d-1)+2 \sqrt{2} a_1
   a_3 \sqrt{d-1} \sqrt{d+2}+3 a_2^2 (d-1)-2 \sqrt{6} a_2 a_3
   \sqrt{d-1} \sqrt{d+2}+2 a_3^2 (d+2)\right),\\
F^{(2)}_{13}&=\frac{1}{18 d}\left(a_1^2 (d-1)+2 a_1 \left(\sqrt{3} a_2 (d-1)+\sqrt{2} a_3
   \sqrt{d-1} \sqrt{d+2}\right)+3 a_2^2 (d-1)+2 \sqrt{6} a_2 a_3
   \sqrt{d-1} \sqrt{d+2}+2 a_3^2 (d+2)\right),\\
F^{(2)}_{14}&=\frac{1}{9d}\left(2 a_1^2 (d-1)-2 \sqrt{2} a_1 a_3 \sqrt{d-1}
   \sqrt{d+2}+a_3^2 (d+2)\right).
\end{split}
\ee
The similar set of equations we can  also obtain  for partition $(1,1)$. Moreover we know that the fidelity from ideal $\mathcal{N}$ is always equal to $1/d$ (see Fact~\ref{fact}). In next step we use \textit{Mathematica} software to generate parametric plots of regions given by formulas of the form~\eqref{form1} together with normalization condition $a_1^2+a_2^2+a_3^2=1$. Thanks to this we get admissible range of fidelities in every irreducible space labeled by partition $\alpha_i$. Due to Theorem~\ref{thm:main} to obtain admissible region of fidelities we have to generate convex hull of allowed regions obtained for every irreducible representation $\alpha$. To do this we have used \textit{Mathematica} package \textit{ConvexHull3D}. One can see that to generate admissible regions for number of clones larger than $3$ we need higher-dimensional space to embed convex hull, so we can not represent our results in the  graphical form. There is still some way to omit this problem at least partially. Namely we can construct some projection which maps  convex hulls from $d-$dimensional space to $3-$dimensional space, but then of course we lose some information.

\addcontentsline{toc}{section}{\bf Bibliography}
\bibliographystyle{apsrev}
\bibliography{mag_PCbib}

\begin{thebibliography}{32}
\expandafter\ifx\csname natexlab\endcsname\relax\def\natexlab#1{#1}\fi
\expandafter\ifx\csname bibnamefont\endcsname\relax
  \def\bibnamefont#1{#1}\fi
\expandafter\ifx\csname bibfnamefont\endcsname\relax
  \def\bibfnamefont#1{#1}\fi
\expandafter\ifx\csname citenamefont\endcsname\relax
  \def\citenamefont#1{#1}\fi
\expandafter\ifx\csname url\endcsname\relax
  \def\url#1{\texttt{#1}}\fi
\expandafter\ifx\csname urlprefix\endcsname\relax\def\urlprefix{URL }\fi
\providecommand{\bibinfo}[2]{#2}
\providecommand{\eprint}[2][]{\url{#2}}

\bibitem[{\citenamefont{Wootters and \.Zurek}(1982)}]{WoottersZurek}
\bibinfo{author}{\bibfnamefont{W.~K.} \bibnamefont{Wootters}} \bibnamefont{and}
  \bibinfo{author}{\bibfnamefont{W.~H.} \bibnamefont{\.Zurek}},
  \bibinfo{journal}{{Nature}} \textbf{\bibinfo{volume}{299}},
  \bibinfo{pages}{802} (\bibinfo{year}{1982}).

\bibitem[{\citenamefont{Dieks}(1982)}]{Dieks}
\bibinfo{author}{\bibfnamefont{D.}~\bibnamefont{Dieks}},
  \bibinfo{journal}{Phys. Lett. A} \textbf{\bibinfo{volume}{92}},
  \bibinfo{pages}{271} (\bibinfo{year}{1982}).

\bibitem[{\citenamefont{Bu\u{z}ek and Hillery}(1996)}]{BuzekHillery}
\bibinfo{author}{\bibfnamefont{V.}~\bibnamefont{Bu\u{z}ek}} \bibnamefont{and}
  \bibinfo{author}{\bibfnamefont{M.}~\bibnamefont{Hillery}},
  \bibinfo{journal}{{Phys. Rev. A}} \textbf{\bibinfo{volume}{54}},
  \bibinfo{pages}{1844} (\bibinfo{year}{1996}).

\bibitem[{\citenamefont{Bru\ss{} et~al.}(1998)\citenamefont{Bru\ss{},
  DiVincenzo, Ekert, Fuchs, Macchiavello, and Smolin}}]{Bruss-cloning1998}
\bibinfo{author}{\bibfnamefont{D.}~\bibnamefont{Bru\ss{}}},
  \bibinfo{author}{\bibfnamefont{D.~P.} \bibnamefont{DiVincenzo}},
  \bibinfo{author}{\bibfnamefont{A.}~\bibnamefont{Ekert}},
  \bibinfo{author}{\bibfnamefont{C.~A.} \bibnamefont{Fuchs}},
  \bibinfo{author}{\bibfnamefont{C.}~\bibnamefont{Macchiavello}},
  \bibnamefont{and} \bibinfo{author}{\bibfnamefont{J.~A.}
  \bibnamefont{Smolin}}, \bibinfo{journal}{{Phys. Rev. A}}
  \textbf{\bibinfo{volume}{57}}, \bibinfo{pages}{2368} (\bibinfo{year}{1998}).

\bibitem[{\citenamefont{Gisin and Massar}(1997)}]{Massar}
\bibinfo{author}{\bibfnamefont{N.}~\bibnamefont{Gisin}} \bibnamefont{and}
  \bibinfo{author}{\bibfnamefont{S.}~\bibnamefont{Massar}},
  \bibinfo{journal}{{Phys. Rev. Lett.}} \textbf{\bibinfo{volume}{79}},
  \bibinfo{pages}{2153} (\bibinfo{year}{1997}).

\bibitem[{\citenamefont{Werner}(1998)}]{Werner-cloning1998}
\bibinfo{author}{\bibfnamefont{R.~F.} \bibnamefont{Werner}},
  \bibinfo{journal}{{Phys. Rev. A}} \textbf{\bibinfo{volume}{58}},
  \bibinfo{pages}{1827} (\bibinfo{year}{1998}).

\bibitem[{\citenamefont{Keyl and Werner}(1999)}]{KW}
\bibinfo{author}{\bibfnamefont{M.}~\bibnamefont{Keyl}} \bibnamefont{and}
  \bibinfo{author}{\bibfnamefont{R.~F.} \bibnamefont{Werner}},
  \bibinfo{journal}{{J. Math. Phys.}} \textbf{\bibinfo{volume}{40}},
  \bibinfo{pages}{3283} (\bibinfo{year}{1999}).

\bibitem[{\citenamefont{{Wang} et~al.}(2011)\citenamefont{{Wang}, {Shi},
  {Xiong}, {Jing}, {Ren}, {Mu}, and {Fan}}}]{Wang-cloning2011}
\bibinfo{author}{\bibfnamefont{Y.-N.} \bibnamefont{{Wang}}},
  \bibinfo{author}{\bibfnamefont{H.-D.} \bibnamefont{{Shi}}},
  \bibinfo{author}{\bibfnamefont{Z.-X.} \bibnamefont{{Xiong}}},
  \bibinfo{author}{\bibfnamefont{L.}~\bibnamefont{{Jing}}},
  \bibinfo{author}{\bibfnamefont{X.-J.} \bibnamefont{{Ren}}},
  \bibinfo{author}{\bibfnamefont{L.-Z.} \bibnamefont{{Mu}}}, \bibnamefont{and}
  \bibinfo{author}{\bibfnamefont{H.}~\bibnamefont{{Fan}}},
  \bibinfo{journal}{{Phys. Rev. A}} \textbf{\bibinfo{volume}{84}},
  \bibinfo{pages}{034302} (\bibinfo{year}{2011}).

\bibitem[{\citenamefont{Braunstein et~al.}(2001)\citenamefont{Braunstein,
  Bu\u{z}ek, and Hillery}}]{Braunstein-cloning2001}
\bibinfo{author}{\bibfnamefont{S.~L.} \bibnamefont{Braunstein}},
  \bibinfo{author}{\bibfnamefont{V.}~\bibnamefont{Bu\u{z}ek}},
  \bibnamefont{and} \bibinfo{author}{\bibfnamefont{M.}~\bibnamefont{Hillery}},
  \bibinfo{journal}{{Phys. Rev. A}} \textbf{\bibinfo{volume}{63}},
  \bibinfo{pages}{052313} (\bibinfo{year}{2001}).

\bibitem[{\citenamefont{Cerf}(2000)}]{Cerf-cloning2000}
\bibinfo{author}{\bibfnamefont{N.~J.} \bibnamefont{Cerf}}, \bibinfo{journal}{J.
  Mod. Opt.} \textbf{\bibinfo{volume}{47}}, \bibinfo{pages}{187}
  (\bibinfo{year}{2000}).

\bibitem[{\citenamefont{{Fiur{\'a}{\v s}ek}
  et~al.}(2005)\citenamefont{{Fiur{\'a}{\v s}ek}, {Filip}, and
  {Cerf}}}]{Fiurasek-cloning2005}
\bibinfo{author}{\bibfnamefont{J.}~\bibnamefont{{Fiur{\'a}{\v s}ek}}},
  \bibinfo{author}{\bibfnamefont{R.}~\bibnamefont{{Filip}}}, \bibnamefont{and}
  \bibinfo{author}{\bibfnamefont{N.~J.} \bibnamefont{{Cerf}}},
  \bibinfo{journal}{Quant. Inform. Comp.} \textbf{\bibinfo{volume}{5}},
  \bibinfo{pages}{583} (\bibinfo{year}{2005}).

\bibitem[{\citenamefont{Iblisdir
  et~al.}(2005{\natexlab{a}})\citenamefont{Iblisdir, Acin, Cerf, Filip,
  Fiur{\'a}{\v s}ek, and Gisin}}]{Iblisdir-cloning2004}
\bibinfo{author}{\bibfnamefont{S.}~\bibnamefont{Iblisdir}},
  \bibinfo{author}{\bibfnamefont{A.}~\bibnamefont{Acin}},
  \bibinfo{author}{\bibfnamefont{N.~J.} \bibnamefont{Cerf}},
  \bibinfo{author}{\bibfnamefont{R.}~\bibnamefont{Filip}},
  \bibinfo{author}{\bibfnamefont{J.}~\bibnamefont{Fiur{\'a}{\v s}ek}},
  \bibnamefont{and} \bibinfo{author}{\bibfnamefont{N.}~\bibnamefont{Gisin}},
  \bibinfo{journal}{Phys. Rev. A} \textbf{\bibinfo{volume}{72}},
  \bibinfo{pages}{042328} (\bibinfo{year}{2005}{\natexlab{a}}).

\bibitem[{\citenamefont{Iblisdir
  et~al.}(2005{\natexlab{b}})\citenamefont{Iblisdir, Acin, and
  Gisin}}]{Iblisdir-cloning2005}
\bibinfo{author}{\bibfnamefont{S.}~\bibnamefont{Iblisdir}},
  \bibinfo{author}{\bibfnamefont{A.}~\bibnamefont{Acin}}, \bibnamefont{and}
  \bibinfo{author}{\bibfnamefont{N.}~\bibnamefont{Gisin}},
  \emph{\bibinfo{title}{{Generalised Asymmetric Quantum Cloning}}}
  (\bibinfo{year}{2005}{\natexlab{b}}), \eprint{arXiv:quant-ph/0505152}.

\bibitem[{\citenamefont{{Kay} et~al.}(2009)\citenamefont{{Kay}, {Kaszlikowski},
  and {Ramanathan}}}]{Kay2009-cloning}
\bibinfo{author}{\bibfnamefont{A.}~\bibnamefont{{Kay}}},
  \bibinfo{author}{\bibfnamefont{D.}~\bibnamefont{{Kaszlikowski}}},
  \bibnamefont{and}
  \bibinfo{author}{\bibfnamefont{R.}~\bibnamefont{{Ramanathan}}},
  \bibinfo{journal}{Phys. Rev. Lett.} \textbf{\bibinfo{volume}{103}},
  \bibinfo{pages}{050501} (\bibinfo{year}{2009}).

\bibitem[{\citenamefont{{Kay} et~al.}(2013)\citenamefont{{Kay}, {Ramanathan},
  and {Kaszlikowski}}}]{Kaszlikowski2012-cloning}
\bibinfo{author}{\bibfnamefont{A.}~\bibnamefont{{Kay}}},
  \bibinfo{author}{\bibfnamefont{R.}~\bibnamefont{{Ramanathan}}},
  \bibnamefont{and}
  \bibinfo{author}{\bibfnamefont{D.}~\bibnamefont{{Kaszlikowski}}},
  \bibinfo{journal}{Quant. Inf. Comp.} \textbf{\bibinfo{volume}{13}},
  \bibinfo{pages}{0880} (\bibinfo{year}{2013}),
  \eprint{arXiv:quant-ph/1208.5574}.

\bibitem[{\citenamefont{{Jiang} and {Yu}}(2010)}]{Yu2010-cloning}
\bibinfo{author}{\bibfnamefont{M.}~\bibnamefont{{Jiang}}} \bibnamefont{and}
  \bibinfo{author}{\bibfnamefont{S.}~\bibnamefont{{Yu}}}, \bibinfo{journal}{J.
  Math. Phys.} \textbf{\bibinfo{volume}{51}}, \bibinfo{pages}{052306}
  (\bibinfo{year}{2010}).

\bibitem[{\citenamefont{Scarani et~al.}(2005)\citenamefont{Scarani, Iblisdir,
  Gisin, and Acin}}]{Gisin}
\bibinfo{author}{\bibfnamefont{V.}~\bibnamefont{Scarani}},
  \bibinfo{author}{\bibfnamefont{S.}~\bibnamefont{Iblisdir}},
  \bibinfo{author}{\bibfnamefont{N.}~\bibnamefont{Gisin}}, \bibnamefont{and}
  \bibinfo{author}{\bibfnamefont{A.}~\bibnamefont{Acin}},
  \bibinfo{journal}{Rev. Mod. Phys.} \textbf{\bibinfo{volume}{77}},
  \bibinfo{pages}{1225} (\bibinfo{year}{2005}).

\bibitem[{\citenamefont{{Fan} et~al.}(2013)\citenamefont{{Fan}, {Wang}, {Jing},
  {Yue}, {Shi}, {Zhang}, and {Mu}}}]{Fan2013-cloning}
\bibinfo{author}{\bibfnamefont{H.}~\bibnamefont{{Fan}}},
  \bibinfo{author}{\bibfnamefont{Y.-N.} \bibnamefont{{Wang}}},
  \bibinfo{author}{\bibfnamefont{L.}~\bibnamefont{{Jing}}},
  \bibinfo{author}{\bibfnamefont{J.-D.} \bibnamefont{{Yue}}},
  \bibinfo{author}{\bibfnamefont{H.-D.} \bibnamefont{{Shi}}},
  \bibinfo{author}{\bibfnamefont{Y.-L.} \bibnamefont{{Zhang}}},
  \bibnamefont{and} \bibinfo{author}{\bibfnamefont{L.-Z.} \bibnamefont{{Mu}}},
  \emph{\bibinfo{title}{{Quantum Cloning Machines and the Applications}}}
  (\bibinfo{year}{2013}), \eprint{arXiv:quant-ph/1301.2956}.

\bibitem[{\citenamefont{{{\'C}wikli{\'n}ski}
  et~al.}(2012)\citenamefont{{{\'C}wikli{\'n}ski}, {Studzi{\'n}ski}, and
  {Horodecki}}}]{Cwiklinski2012-cloning}
\bibinfo{author}{\bibfnamefont{P.}~\bibnamefont{{{\'C}wikli{\'n}ski}}},
  \bibinfo{author}{\bibfnamefont{M.}~\bibnamefont{{Studzi{\'n}ski}}},
  \bibnamefont{and}
  \bibinfo{author}{\bibfnamefont{M.}~\bibnamefont{{Horodecki}}},
  \bibinfo{journal}{Phys. Lett. A.} \textbf{\bibinfo{volume}{32}},
  \bibinfo{pages}{2178} (\bibinfo{year}{2012}).

\bibitem[{\citenamefont{{Studzi{\'n}ski}
  et~al.}(2013)\citenamefont{{Studzi{\'n}ski}, {Horodecki}, and
  {Mozrzymas}}}]{Studzinski2013}
\bibinfo{author}{\bibfnamefont{M.}~\bibnamefont{{Studzi{\'n}ski}}},
  \bibinfo{author}{\bibfnamefont{M.}~\bibnamefont{{Horodecki}}},
  \bibnamefont{and}
  \bibinfo{author}{\bibfnamefont{M.}~\bibnamefont{{Mozrzymas}}},
  \bibinfo{journal}{J. Phys. A: Math. Theor. 46 (2013) 395303 (31pp)}
  (\bibinfo{year}{2013}), \eprint{arXiv:quant-ph/1305.6183}.

\bibitem[{\citenamefont{Mozrzymas et~al.}(2014)\citenamefont{Mozrzymas,
  Horodecki, and Studzi{\'n}ski}}]{mozrzymas2013}
\bibinfo{author}{\bibfnamefont{M.}~\bibnamefont{Mozrzymas}},
  \bibinfo{author}{\bibfnamefont{M.}~\bibnamefont{Horodecki}},
  \bibnamefont{and}
  \bibinfo{author}{\bibfnamefont{M.}~\bibnamefont{Studzi{\'n}ski}},
  \bibinfo{journal}{JMP} \textbf{\bibinfo{volume}{55}}, \bibinfo{pages}{032202}
  (\bibinfo{year}{2014}).

\bibitem[{\citenamefont{Eggeling and Werner}(2001)}]{EggelingWerner}
\bibinfo{author}{\bibfnamefont{T.}~\bibnamefont{Eggeling}} \bibnamefont{and}
  \bibinfo{author}{\bibfnamefont{R.~F.} \bibnamefont{Werner}},
  \bibinfo{journal}{{Phys. Rev. A}} \textbf{\bibinfo{volume}{63}},
  \bibinfo{pages}{2111} (\bibinfo{year}{2001}).

\bibitem[{\citenamefont{Eggeling}(2003)}]{Eggeling}
\bibinfo{author}{\bibfnamefont{T.}~\bibnamefont{Eggeling}}, Ph.D. thesis,
  \bibinfo{school}{Braunschweig} (\bibinfo{year}{2003}).

\bibitem[{\citenamefont{Horodecki et~al.}(1999)\citenamefont{Horodecki,
  Horodecki, and Horodecki}}]{MHPH}
\bibinfo{author}{\bibfnamefont{M.}~\bibnamefont{Horodecki}},
  \bibinfo{author}{\bibfnamefont{P.}~\bibnamefont{Horodecki}},
  \bibnamefont{and}
  \bibinfo{author}{\bibfnamefont{R.}~\bibnamefont{Horodecki}},
  \bibinfo{journal}{{Phys. Rev. A}} \textbf{\bibinfo{volume}{60}},
  \bibinfo{pages}{1888} (\bibinfo{year}{1999}).

\bibitem[{\citenamefont{Horodecki}(2001)}]{Horodecki_2001}
\bibinfo{author}{\bibfnamefont{M.}~\bibnamefont{Horodecki}},
  \bibinfo{journal}{Quantum Info. Comput.} \textbf{\bibinfo{volume}{1}},
  \bibinfo{pages}{3} (\bibinfo{year}{2001}).

\bibitem[{\citenamefont{{Zhang} et~al.}(2006)\citenamefont{{Zhang}, {Kauffman},
  and {Werner}}}]{Werner2006}
\bibinfo{author}{\bibfnamefont{Y.}~\bibnamefont{{Zhang}}},
  \bibinfo{author}{\bibfnamefont{L.~H.} \bibnamefont{{Kauffman}}},
  \bibnamefont{and} \bibinfo{author}{\bibfnamefont{R.~F.}
  \bibnamefont{{Werner}}}, \bibinfo{journal}{Int.J.Quant.Inf.}
  \textbf{\bibinfo{volume}{5}}, \bibinfo{pages}{469} (\bibinfo{year}{2006}),
  \eprint{quant-ph/0606005}.

\bibitem[{\citenamefont{Choi}(1972)}]{Choi1972}
\bibinfo{author}{\bibfnamefont{M.-D.} \bibnamefont{Choi}},
  \bibinfo{journal}{Can. J. Math.} \textbf{\bibinfo{volume}{3}},
  \bibinfo{pages}{520} (\bibinfo{year}{1972}).

\bibitem[{\citenamefont{Jamio\l{}kowski}(1972)}]{Jamiolkowski}
\bibinfo{author}{\bibfnamefont{A.}~\bibnamefont{Jamio\l{}kowski}},
  \bibinfo{journal}{{Rep. Math. Phys.}} \textbf{\bibinfo{volume}{3}},
  \bibinfo{pages}{275} (\bibinfo{year}{1972}).

\bibitem[{\citenamefont{Lieb and Mattis}(1962)}]{Lieb1962-LB}
\bibinfo{author}{\bibfnamefont{E.}~\bibnamefont{Lieb}} \bibnamefont{and}
  \bibinfo{author}{\bibfnamefont{D.}~\bibnamefont{Mattis}},
  \bibinfo{journal}{Journal of Mathematical Physics}
  \textbf{\bibinfo{volume}{3}}, \bibinfo{pages}{749} (\bibinfo{year}{1962}).

\bibitem[{\citenamefont{{Nachtergaele} and
  {Starr}}(2005)}]{Nachtergaele2006-LM}
\bibinfo{author}{\bibfnamefont{B.}~\bibnamefont{{Nachtergaele}}}
  \bibnamefont{and} \bibinfo{author}{\bibfnamefont{S.}~\bibnamefont{{Starr}}},
  \bibinfo{journal}{Physical Review Letters} \textbf{\bibinfo{volume}{94}},
  \bibinfo{pages}{057206} (\bibinfo{year}{2005}),
  \eprint{arXiv:math-ph/0408020}.

\bibitem[{\citenamefont{Fulton and Harris}(1991)}]{Fulton1991-book-rep}
\bibinfo{author}{\bibfnamefont{W.}~\bibnamefont{Fulton}} \bibnamefont{and}
  \bibinfo{author}{\bibfnamefont{J.}~\bibnamefont{Harris}},
  \emph{\bibinfo{title}{Representation Theory - A first Course}}
  (\bibinfo{publisher}{Springer-Verlag}, \bibinfo{address}{New York},
  \bibinfo{year}{1991}).

\bibitem[{\citenamefont{Curtis and Reiner}(1988)}]{Curtis}
\bibinfo{author}{\bibfnamefont{C.~W.} \bibnamefont{Curtis}} \bibnamefont{and}
  \bibinfo{author}{\bibfnamefont{I.}~\bibnamefont{Reiner}},
  \emph{\bibinfo{title}{Representation Theory of Finite Groups and Associative
  Algebras}} (\bibinfo{publisher}{John Wiley and Sons}, \bibinfo{address}{New
  York}, \bibinfo{year}{1988}).

\end{thebibliography}
\end{document}